\documentclass[sigconf,screen]{acmart}
\usepackage{hyperref}
\AtBeginDocument{%
  }

\copyrightyear{2026}
\acmYear{2026}
\setcopyright{cc}
\setcctype{by}
\acmConference[SPAA '26]{38th ACM Symposium on Parallelism in Algorithms and Architectures}{July 06--10, 2026}{London, United Kingdom}
\acmBooktitle{38th ACM Symposium on Parallelism in Algorithms and Architectures (SPAA '26), July 06--10, 2026, London, United Kingdom}
\acmDOI{10.1145/3816782.3819198}
\acmISBN{979-8-4007-2761-0/2026/07}


\usepackage[T1]{fontenc}
\usepackage[utf8]{inputenc}

\usepackage{listings, fancyvrb}
\usepackage{pifont}
\usepackage{amsmath, amsthm}
\usepackage[inline]{enumitem}
\usepackage{nicefrac}       % compact symbols for 1/2, etc.
\usepackage{graphicx}
\usepackage{subcaption}
\usepackage{booktabs}  % For professional tables
\usepackage{adjustbox} % For image alignment
\usepackage{array}
\usepackage{makecell}
\usepackage{textgreek}

\usepackage{wrapfig}

\usepackage[toc,page]{appendix}

\usepackage{microtype}
\newif\ifdraft
\draftfalse %ready mode

\newcommand{\myparagraph}[1]{\vspace{.03in}\noindent {\bf #1.}}
\usepackage[framemethod=TikZ]{mdframed}
\mdfdefinestyle{listingstyle}{%
  outerlinewidth=0.25pt,outerlinecolor=black,%
  innerleftmargin=5pt,innerrightmargin=5pt,innertopmargin=0pt,innerbottommargin=0pt%
}

\makeatletter
\renewenvironment{proof}[1][\proofname]{%
  \par\pushQED{\qed}%
  \normalfont
  \topsep=0pt
  \partopsep=0pt
  \trivlist
  \item[\hskip\labelsep
        \itshape
    #1\@addpunct{.}]\ignorespaces
}{%
  \popQED\endtrivlist\@endpefalse
}
\makeatother

\usepackage[T1]{fontenc}
\usepackage{caption}
\usepackage{float}
\ifdraft
        \newcommand{\Zak}[1]{{\color{red} Zak: #1}}
        \newcommand{\Andre}[1]{{\color{blue} Andre: #1}}
        \newcommand{\Guy}[1]{{\color{green} Guy: #1}}
        \newcommand{\hedit}[1]{{#1}}
\else
        \newcommand{\Zak}[1]{}
        \newcommand{\Andre}[1]{}
        \newcommand{\Guy}[1]{}
        \newcommand{\hedit}[1]{}
\fi

\newcommand{\ustm}{\textsc{\textmu STM}}
\newcommand{\twoplsf}{2PLSF}
\newcommand{\fuse}{\textsc{fuse}}
\newcommand{\multiverse}{\textsc{Multiverse}}

\newcommand{\future}[1]{} % things to do in the future

\makeatletter
\lst@Key{countblanklines}{true}[t]%
{\lstKV@SetIf{#1}\lst@ifcountblanklines}

\lst@AddToHook{OnEmptyLine}{%
	\lst@ifnumberblanklines\else%
	\lst@ifcountblanklines\else%
	\advance\c@lstnumber-\@ne\relax%
	\fi%
	\fi}
\lst@AddToHook{OnEmptyLine}{\vspace{-0.3\baselineskip}}
\makeatother

\newtheorem{theorem}{Theorem}[section]

\author{Zachary Kent}
\affiliation{%
        \institution{Carnegie Mellon University}
	\country{Pittsburgh, PA, USA}
}
\email{zkent@cs.cmu.edu}

\author{Guy E. Blelloch}
\affiliation{%
         \institution{Carnegie Mellon University}
  	\country{Pittsburgh, PA, USA}
}
\email{guyb@cs.cmu.edu}

\author{Andr\'{e} Costa}
\affiliation{%
        \institution{Carnegie Mellon University}
	\country{Pittsburgh, PA, USA}
}
\email{ajcosta@cs.cmu.edu}


\begin{document}
\title[\ustm: A Lightweight and Efficient STM]{
  \ustm: A Lightweight and Efficient STM Supporting\\ General Types and Deferred Aborts
}

\begin{abstract}
Software Transactional Memory (STM) systems allow developers to more easily exploit multicore architectures by wrapping arbitrary sequential code in \textit{transactions} that are executed concurrently. In recent years, the performance of STM systems has approached that of hand-tuned data structures through techniques that avoid unnecessary aborts and exploit the semantics of underlying data structures. 

Despite achieving excellent performance, most STM systems do not fully
address the concerns they targeted in the first place: safety,
usability, and generality.  In particular, these systems place
restrictions on the data types that may be updated transactionally,
such as requiring that these types fit within a word, and can require
modification of data layout.  Moreover, most STM systems abort
transactions in the middle of client code to ensure
correctness.  This can cause space leaks and other bugs not present in
the original code.

We present \ustm{}, a novel STM system addressing all of these
shortcomings while still maintaining excellent performance, all within
$\sim$300 lines of code. \ustm{} supports general types while maintaining
data layout.  Aborts are \textit{deferred} until the end of the
transaction, allowing client code within a transaction to terminate
normally. To ensure that \ustm{} guarantees opacity, we implement a novel timestamping algorithm we call \textit{split-increment timestamps}.

We compare the performance of \ustm{} to a variety of state-of-the-art (SOTA) STM systems, demonstrating that \ustm{} matches or outperforms the SOTA on a
variety of workloads.
\end{abstract}

\maketitle
\section{Introduction}

Transactional memory (TM) allows developers to wrap code in
transactions such that all accesses to shared memory within the code
appear as if they happen atomically even when other threads are
concurrently accessing the memory.  TM has long been suggested as a
way to greatly simplify concurrent programming on shared-memory
multicore architectures.  However, early libraries for transactional
memory did not perform well and had various restrictions that made
them difficult to use in practice~\cite{CBMCWCC08}.  Over the
years there have been many advances that have improved performance.
Such improvements include removing levels of indirection, supporting
opacity~\cite{tl2}, more efficient locks~\cite{2plsf}, better
contention management~\cite{contention09}, efficient
multiversioning~\cite{lu2013generic,neumann2015fast,Perelman11,wu17},
efficient timestamping~\cite{tl2,RC24}, and efficient memory
management~\cite{lu2013generic,WBFR23}.  Software Transactional
Memory (STM) libraries are in many cases quite efficient, and
experiments have shown that data structures based on STM libraries can
approach the efficiency of hand-coded concurrent data structures in
several situations~\cite{RC24,tlf25}.

Such modern STM libraries, however, still have notable limitations
with regards to general usage.  Limitations include requiring
indirection for anything other than trivial types, relying on unsafe
long jumps (unsafe aborts), requiring modification of the underlying
data structures (intrusive), performing badly under high contention
(contention intolerant), and requiring all reads and writes to STM
variables to be in a transaction (no publication or privatization).
Table~\ref{table:others} summarizes these limitations across a variety
of state-of-the-art STM libraries.  More details on these limitations
and implications are given in Section~\ref{sec:overview}.

\newcommand{\tc}{\ding{52}}
\newcommand{\tx}{\ding{56}}
\newcommand{\tq}{\ding{182}}

\begin{table*}
  \begin{tabular}{|l|c|c|c|c|c|c|c|}
    \hline
    System &
    \parbox{.5in}{\center General Types} &
    \parbox{.5in}{\center Safe Aborts} &
    \parbox{.5in}{\center Non-intrusive} &
    \parbox{.6in}{\center Contention Tolerant} &
    \parbox{.7in}{\center Privatization Safety} &
    \parbox{.6in}{\center Multi\\ Versioned} &
    \parbox{.6in}{\center Starvation Free} \\ \hline
        
%                  & gt  & sa  & ni  & ct  & ps  & mv  & sf  \\ \hline
    \sc TinySTM~\cite{tiny08}    & \tx & \tx & \tx & \tx & \tx\rlap{$^2$} & \tx & \tx \\ \hline
    \sc tl2~\cite{tl2}        & \tx & \tx & \tx & \tx & \tx\rlap{$^2$} & \tx & \tx \\ \hline
    \sc \twoplsf{}~\cite{2plsf}   & \tx & \tx & \tx & \tx & \tc & \tx & \tc \\ \hline
    \sc dctl~\cite{RC24}      & \tx & \tx & \tx & \tx\rlap{$^1$} & \tx\rlap{$^2$} & \tx & \tx\rlap{$^3$}  \\ \hline
    \sc multiverse~\cite{Coccimiglio0R26} & \tx & \tx & \tx & \tx & \tx\rlap{$^2$} & \tc & \tx \\ \hline
    \sc \fuse{}~\cite{tlf25}       & \tx & \tc\rlap{$^{4}$} & \tx & \tc & \tc & \tc & \tx \\ \hline
    \sc\bf \ustm{} (this paper)   & \tc & \tc & \tc & \tc & \tc & \tc & \tx \\ \hline
  \end{tabular}
  \vspace{.1in}
  \caption{Properties of various STM systems.  All these systems
    support opaque serializable transactions, and support allocation
    and frees within a transaction. The last three are all from within
    the past three years.
    (1) The contention status of DCTL is unknown, as its implementation is proprietary.
    (2) We discuss how these STMs could efficiently support privatization
    safety in the full version of this paper~\cite{fullpaper}.
    (3) DCTL has a version that is starvation free.
    (4) For safe aborts \fuse{} requires hardware timestamps.
    }
  \label{table:others}
  \vspace{-15pt}
\end{table*}
In this paper we present \ustm, a header-only STM library addressing each of these limitations.
For this purpose we introduce new techniques, including
\emph{deferred aborts}, \emph{split-increment timestamps}, and
\emph{allocate-swap-retire}.
\ustm{} uses multiversioning~\cite{Reed78} and optimistic concurrency
control~\cite{KungR81}, where a speculative phase runs the user code
buffering any writes, and a commit phase validates and executes the
writes if successful.
Our experiments show that \ustm{} is the fastest publicly available
STM across a broad set of workloads.
%, including workloads under high contention.
Furthermore our STM consists of only around 300 lines of pure C++
code (measured by \texttt{cloc}), plus another 150 lines for an epoch-based reclamation
scheme\footnote{\url{https://github.com/cmuparlay/ustm}}.  \ustm{}
should therefore be reasonably easy to adapt and extend.  Here we
briefly go through the techniques we introduce and how they help
alleviate the limitations; more details are given in Section~\ref{sec:overview}.

Our deferred aborts are used to avoid any exceptional control flow in
user code.  Most prior STMs rely on unsafe system longjumps to exit
user code.  These longjumps, or other exceptional control, are
necessary in the systems to ensure opacity~\cite{opacity08} (i.e.,
even aborted transactions must see a snapshot of the state).  Instead,
our deferred aborts continue running the user code until finished, and
then abort, if needed, before committing any writes.  The challenge is
to maintain opacity.  We support this with a combination of
multiversioning and a new form of efficient timestamping which we
refer to as split-increment timestamps.  These timestamps solve a
problem with prior methods for efficient timestamps, e.g. used in
TL2~\cite{tl2}, Verlib~\cite{BlellochW24}, and DCTL~\cite{RC24}, while
maintaining their efficiency.  In particular, these prior methods do
not allow deferred aborts since they cannot capture a snapshot for
aborted transactions even with multiversioning.

%% The immediate retire is safe since in
%% conjunction with safe memory reclamation the cell will not be
%% freed until all transactions that are live at the time of the retire,
%% and hence might need an old version, have completed.  

Our allocate-swap-retire approach is used for two purposes: supporting
multiversioning and non-intrusive indirection-free general types.  The
idea of the approach is that during a transaction when a value is
stored, we allocate a cell to hold it and add a pointer to the cell to
a per-transaction write log.  Later when committing the transaction we
swap the value in the cell with the value in the location, immediately
retire the cell, and then link the cell into a version list, which is
stored separately in a hash table.  A key efficiency is achieved by
using the same mechanism for supporting general types and
multiversioning---multiversioned STMs have to allocate such cells in
any case.  Importantly, the approach is non-intrusive since we just
use the original layout of the data and all meta information (locks
and version lists) is kept separately.  Care is required to
safely read locations.

%% Reads in a transaction can
%% access the value directly, although if in a transaction they do have
%% to check for consistency in the hash table.

Beyond the above-mentioned approaches \ustm{} uses variants on
standard approaches to achieve functionality and performance.  It
supports publication and privatization safety using fences~\cite{safeprivatization},
supporting both global and per-variable fences (see the full version~\cite{fullpaper}).
%% Importantly since
%% \ustm{} is non-intrusive, when the structure is accessed outside of
%% transactions, and with appropriate publication and privatization
%% fences, it can be accessed as a standard sequential structure.
With regards to reducing contention costs,
our allocate-swap-retire and split-increment timestamps both 
improve contention costs, as discussed later.
% On the other hand, \ustm{} is not starvation free.

We have implemented \ustm{} in pure C++, and
make no use of longjumps or exceptions.  We have tested
portability across x86, ARM and PowerPC processors.
\ustm{} was designed to be easily incorporated into
existing C++ code and extended.  Its public interface is given in
the full version~\cite{fullpaper}.

In the experimental section we compare performance to several existing
STM systems.  We compare on eight different data structures supporting
dictionaries using YCSB-like workloads, and on the widely used TPC-C
benchmark suite.  The workloads vary table size, update rates,
contention (using a zipfian distribution), transaction size, and thread
counts.  We report both on geometric mean across workloads, and graphs
for each parameter.  To see the advantage of multiversioning we also
report on range query performance.  We outperform other systems in
almost all scenarios.

To better understand the effect of our split-increment timestamps we
compare across a variety of timestamp approaches including hardware
stamps (only available on Intel x86 machines), eager stamps, and lazy
stamps (only safe with longjumps).  To better understand the cost of
our deferred aborts, and allocate-swap-retire technique, we do
an ablation study where we compare two versions of \ustm{} that use
longjumps instead of deferred aborts, and one that additionally only
supports single-version trivial types and does not need any allocation
on writes.  These stripped down versions do perform slightly better on
the structures they can deal with, but with much less functionality.

\section{Overview and Related Work}
\label{sec:overview}

In this section we describe why the features we support are important, present more
detail on how we address the issue, and how it relates to prior work.

\subsection{Multiversion and Optimistic Transactions}

Multiversion transactional systems date back to the
1970s~\cite{Reed78} and are widely used in both transactional database
systems and transactional
memory~\cite{Reed78,BG83,papadimitriou1984concurrency,perelman2010maintaining,Postgres12,SQL13,Kumar14,neumann2015fast,wu17,LKA17,riegel2006lazy,cachopo2006versioned,DieguesR15,Perelman11,FernandesC11,Coccimiglio0R26,tlf25}.
Multiversioning keeps old versions of values when overwritten,
typically in per-location version lists, so ongoing transactions can
make use of them.  Their main advantage is allowing large read-only
transactions to proceed concurrently with updating transactions, and in
several of these systems the read-only transactions never abort.  We also
use multiversioning to support deferred aborts.

Optimistic concurrency control dates back to the early
1980s~\cite{KungR81} and is used in the majority of both transactional
database systems and transactional memory systems (a citation list
would be too long but Harris, Larus and Rajwar describe
several such systems~\cite{HLR10}).  The idea is to split a transaction into a
speculative phase that runs the user code and a commit phase to check
for consistency by validating the reads, and if successful committing
the writes.  This compares to pessimistic concurrency, which takes
locks on all memory accesses during the user code and never has to
validate.  The advantage of optimistic concurrency is that it does
not require any locks on the reads.  However it is the lack of
read locks that makes handling non-trivial types difficult.
 
\subsection{General Types and Non-intrusion}

General-types and non-intrusion are two related but orthogonal concepts
requiring that the STM system support transactions over general types without
data-layout changes to these types. Most existing STMs fail these requirements.
For example, many require that the underlying type fit into a word, or
further that the type itself be a pointer (\texttt{T*} for some type \texttt{T}).
Some STMs further require \emph{intrusive data layout changes} to this type \texttt{T},
requiring that it inherit from a base class provided by the STM. For example, \fuse{}
requires that types \texttt{T} inherit from the \texttt{fuse::versioned} struct.

These characteristics have implications for both performance and usability.
In particular, we are concerned with transactions that access inlined objects,
i.e., objects that are placed adjacently in memory within the
containing structure.  Boxed objects, on the other hand, are ones that
are stored in the heap with a pointer to them in the containing
structure.  In languages such as C++ and Rust all objects are inlined,
while in other languages, such as Java and Haskell, the compiler
decides when to inline~\cite{boxing95,Olsson2022,HengleinJ94}.  Inlining can
greatly reduce space and time, since it avoids allocations and a level
of indirection, which can often incur a cache miss.  Consider, for
example, an array of 16 4-byte objects.  Assuming 64 byte cache lines,
this would require a single cache line if stored inlined, but
significantly more if stored boxed.  Furthermore, scanning the array
values would touch one cache line inlined instead of 17 if boxed.

Boxed objects are relatively easy to support in an STM if the
STM supports memory allocation within a transaction.  This is because
pointers fit in a single word and can be atomically updated with
hardware instructions.  Hence general types can be supported in most
STMs, by creating a wrapper that boxes objects.  The wrapper would
convert an inlined store of an object, for example, to an indirect
one, by allocating a ``box'' in which to put the new object, reading
the pointer to the current box and retiring it, and writing in the
pointer to the new box.  The boxing, however, will incur significant
cost both for loads (reading indirectly) and stores.  Furthermore
boxing objects would be intrusive requiring changes to the existing
data layout.  Unless integrated with a compiler, these changes are likely to
make the structure incompatible with any existing
code that uses it.

Unlike boxed objects, supporting non-trivial inlined objects in STMs
can be tricky, especially with optimistic concurrency control.  Much
of the efficiency of optimistic STMs comes from avoiding taking locks
when reading.  Instead they detect read-write conflicts later during
an opacity check or during validation.  For non-trivial types this can
lead to situations where the readers see partially written values.
Another issue is that the write logs need to store arbitrary objects
so that writes can be buffered.

Our allocate-swap-retire approach aims to support inlined objects for
most types in a way that can be used safely in an optimistic
STM. Furthermore, the approach also supports multiversioning with the
same mechanism, avoiding a double cost.  The approach requires that
the type is copyable and relocatable.  This second condition has been
discussed in the C++
community~\cite{relocatableODwyer24,relocatableMeredith25} and
effectively means that two inlined objects can be swapped or moved by just
swapping or moving their bytes.  This is true for most implementations of C++
classes, including e.g., std::vector.

Allocate-swap-retire works, roughly, as follows.  Each store
in user code, which is run during the speculative phase, allocates an
object to hold a copy of the stored value, and adds a pointer to the
object to the write log.  Later, during the
commit phase, if the transaction succeeds, we use a bytewise swap of
the current value (in the location) and the new value (in the
allocated object).  The object is then tagged with the
version number of the transaction, added to a version list for the
location, and retired.  To be non-intrusive, we store the version lists elsewhere
by keeping a fixed number of buckets, and hashing the location's
address to one of these buckets to store the version~\cite{lu2013generic}.
This is also where we store locks.
A version list can therefore be associated with multiple location
addresses.  We prove that the immediate retiring is safe
(Theorem~\ref{theorem:memsafe}).
Since the cells are
short-lived, they are recycled quickly and are ``warm'' in the cache
for reuse (assuming a decent memory allocator with thread local
pools).  

On a load, the allocate-swap-retire idiom needs to properly load the
value even though another thread could be concurrently updating it.
To implement this we use a variant of sequence locks~\cite{Hemminger02,Lameter05,Boehm12,Sullivan17} and can
take advantage of the timestamp already used for multiversioning.  In
particular the read first reads the timestamp from the head of the
hashed version list, copies the bytes of the type to a buffer,
and then reads the timestamp again.  It checks that the two timestamps
are equal, are not locked, and are less than the start stamp
associated with the transaction.
If so it has properly read the bytes and it can now copy the value out
of the buffer (using e.g., a copy constructor in C++) and return the
copy.  Note this is where being relocatable is important since it
must be the case that the copy constructor acts equivalently whether
the value is in its original location or the buffer.  If the
timestamps are not equal or locked we repeat.  If the timestamps are
not less than the start stamp, then we traverse the version list
searching for the correct version.  The details of this mechanism, and
how it handles the other cases are described in Section~\ref{sec:algorithm}.

\subsection{Deferred Aborts}

Most TM systems we know of use longjumps within user code to implement
aborts (e.g., \twoplsf{}~\cite{2plsf}, TinySTM~\cite{tiny08}, TL2~\cite{tl2}, DCTL~\cite{RC24}, Multiverse~\cite{Coccimiglio0R26},
Trinity~\cite{ramalhete2019onefile}, tl4x~\cite{ACRSF23}). Some form
of exceptional control flow (either long jumps, exceptions, or having
users thread the errors themselves) is required to support
opacity~\cite{opacity08} in single-version systems.  This is because
the user might load a variable that has been updated since its
transaction started and hence be inconsistent with prior reads.  In
principle this problem can be alleviated in a multiversion system
since the load could retrieve the value valid at the start of the
transaction presenting user code with a snapshot of the state.  In
practice, however, this is more difficult, as discussed below, but
let's start with why longjumps and other exceptional control flow are
bad for general use.

A longjump~\cite{Clanguage} works, roughly, by saving the register
state at a given point in the code and then allowing the user to
``jump'' back to that point by restoring that state.  The jump could
pop up many layers of function calls.  This is extremely dangerous in
the RAII style of programming~\cite{raii} of C++ or Rust since none of
the destructors on the stack will be called, potentially leaking
memory, leaving streams unclosed, or locks held (although we hope
users do not put locks in a transaction).  Indeed we found that our
B-tree code had a memory leak when used with many of the systems we
experimented with since the constructor for a node copied from another
node using transactional loads.  If one of these loads aborted, the
memory for the new node would not be collected.  Even more dangerous,
and used by \twoplsf{}, is to add the object to a retire-on-abort list
before calling the constructor.  This would destruct the node on
abort, but the node could only be partially filled when it takes the
longjump so the destructor is later applied to an inconsistent state.

Using exceptions in C++ (or other languages) is much safer since they
``unravel'' the stack when an exception is thrown, applying all
destructors on the way up the stack to the catch point.  Exceptions,
however, have the opposite problem---to perform correctly they require
judicious use of RAII programming.  This requires, for example,
replacing all raw pointers with smart pointers.  Additionally, since
exceptions are designed for uncommon cases, they are expensive when
actually thrown.  However, aborts are not necessarily exceptional---in
some of our high-contention benchmarks we get 30x more aborts than
successes.  Furthermore, replacing pointers with smart shared pointers
can be extremely costly in a concurrent environment since concurrent
reads would contend on incrementing the reference
counter~\cite{AndersonBW21}.  We have hence never seen data structures
designed for transactional memory that use smart pointers.
The last option is to have users thread the errors ``up the stack'',
but this is also not a satisfactory solution.
%% themselves, tagging all return values as a possible error, and running
%% a test after every function call to check if the result is an error
%% and immediately return the error up one more level if so.  Since rust
%% and Haskell do not have exceptions they use this approach and add
%% syntactic sugar and combinators to make it easier.  However, the
%% approach can require significant restructuring of code, and extremely
%% messy without the proper support.

\subsection{Split Timestamps}

All methods to abort user code in the middle have
significant problems, at least for general use.  \ustm{} therefore
runs user code to completion, but, as mentioned, this requires that
the user sees a consistent snapshot even if it aborts, requiring, at
least, multiversioning.  The problem is that all the multiversioning
systems require maintaining timestamps.  As has been noted by many,
incrementing timestamps on every transaction is prohibitively
expensive~\cite{tl2,RC24,BlellochW24,LKA17,YPSD16,wu17}.  Therefore
all practical systems we are aware of use some form of lazy or
imprecise timestamp~\cite{tl2,BlellochW24,RC24,LKA17,YPSD16}.

\begin{figure}
  \includegraphics[width=\columnwidth]{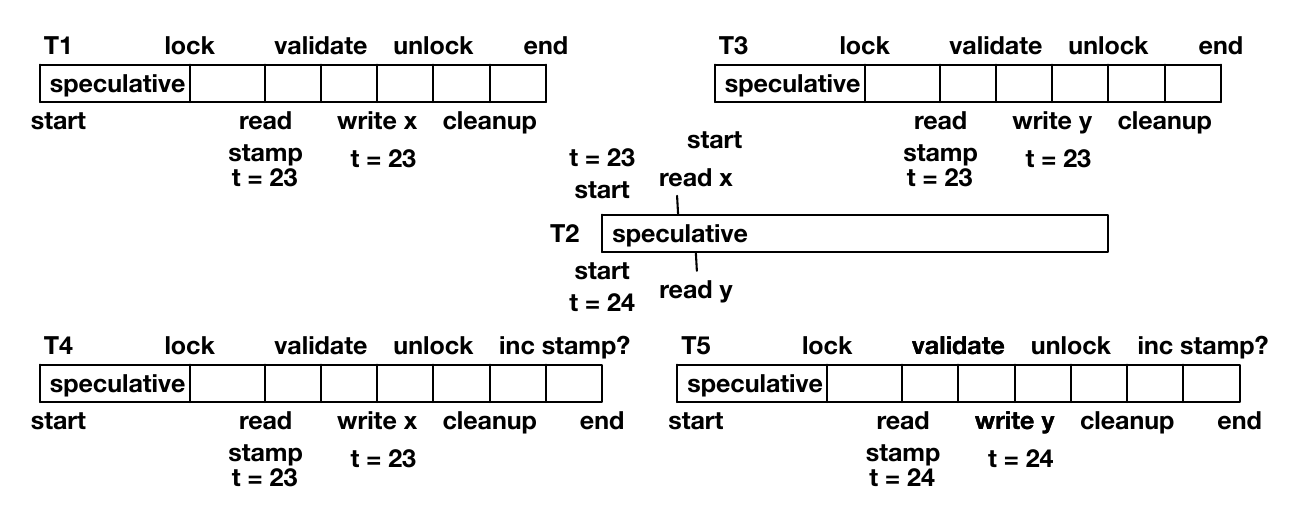}
  \caption{Example of lazy vs. split-increment stamps.  Time from left to right. T1 and T3 use
    lazy stamps.  With them there is no way for T2, which starts at $t = 23$, to determine which
    version of $x$ to read.   To be strictly serializable it must read from T1, but
    it cannot see the $y$ from T3 since it has not yet happened, but the write of $y$ happened at the ``same time'' as the write to $x$.  To be opaque, T2 must abort on reading $x$, increment the stamp, and restart.   T4 and T5 use split-increment stamps.  In this case T2 starts at $t = 24$
    and it is safe for it to use the value of $x$ from T4.  T4 will not need to increment the stamp if another transaction has incremented it since its read stamp.}
  \label{fig:timestamp}
\end{figure}

The idea of a lazy timestamp~\cite{tl2,BlellochW24,RC24} is that
timestamps are not incremented when the transaction is successful, but
are when they fail (in the case of TL2~\cite{tl2} they are sometimes
incremented when successful).  Instead, these systems detect when
reading a value that the timestamp ordering is ambiguous---in
particular that there is no way to properly order an update that is
read relative to the ongoing transaction.  If such an ambiguous
ordering is detected, the transaction must abort immediately (i.e.,
with exceptional control flow) to preserve opacity.  This is true even
in a multiversion TM since the system cannot decide which version to
use.  This problem also occurs with imprecise
timestamps~\cite{LKA17,YPSD16} and seems inherent with all relaxed
timestamp approaches.  In addition to forcing exceptional control flow
(e.g., a longjump) in the middle of user code, it can force read-only
transactions to abort even when using multiversioning.

We introduce split-increment timestamps to avoid this problem.  As
with lazy stamps, they typically avoid increments, but they ensure
that the ordering of a read is never ambiguous.  This allows the
system to defer the aborts, and also avoids any aborts on read-only
transactions.  The idea is to read the stamp early during the commit
phase, and then increment it at the very end of the transaction (after
all locks are released and cleanup is complete), but only if it has not
been incremented by another transaction in the meantime.  Under high
contention on the clock most transactions do not need to do the
increment since some other thread has incremented the stamp in the
meantime.  The correctness is subtle.  We prove that this is safe
(Section~\ref{sec:correctness}) and show experimentally that it is
efficient---not quite as efficient as lazy stamps, but much more
efficient than eager stamps.   Figure~\ref{fig:timestamp} illustrates
the problem with lazy stamps, and how split-increment stamps avoid
the problem.

\subsection{Privatization and Contention}

An issue that is understood in the
literature~\cite{SMDS07,dice2010implicit,safeprivatization}, but not
commonly addressed by existing STM systems is the interoperability of
transactional and non-transactional code.  In particular, user
programs might require that variables previously accessed inside a
transaction be used outside of the STM context (privatization) or vice
versa (publication).  The question is how can an STM provide
privatization (and publication) safety.  This is more of a problem
with optimistic systems than pessimistic ones~\cite{safeprivatization}.
Based on ideas of Khyzha, Attiya, Gotsman and Rinetzky
(KAGR)~\cite{safeprivatization} we supply fence operations.  In
addition to a global fence suggested by KAGR, we supply a per location
fence.  This is discussed further in the full version~\cite{fullpaper}.

With regards to contention, there are several features of \ustm{} that
are designed to improve performance under high contention.  Firstly
\ustm{} aims to minimize the work that is performed in the critical
region in the commit phase when locks are taken.  Under high
contention, the critical regions sequentialize and hence reducing the
time in the region reduces the critical path of the computation.  To
this end, we ensure that no memory management is performed in
the critical region.  In the allocate-swap-retire approach the
allocate is performed before the critical region and the retire after.
All user allocations and deletes are performed in the speculative phase.
Also, with split-increment timestamps any increments of the stamp are
performed outside of the critical region.  Within the critical region
we only read the stamp.

Secondly, we use try locks with early validates and aborts.  In
particular, before even trying to take a lock we check that the
location is still valid and abort if not.  Although not strictly
necessary for correctness, in practice most locations that are written
are also read.  This means that if the validate on a write location
fails, the transaction is most likely to abort during the read validations.
Hence, taking the lock was a waste, possibly delaying other threads.
Using try locks instead of strict locks has a similar benefit.  If a
lock is busy when encountered, the transaction with the lock will
update the location.  Hence the transaction that sees the busy lock
will, again, likely abort due to a validation failure on the location.

%% \subsection{Minimally Intrusive Instrumentation}

%% We would like to minimize changes to user code.  In particular we don't want
%% to modify any types on the data structures, and ideally just have
%% to annotate any transactional reads or writes with a transactional version.

%% Consider, for example, a graph data structures and an algorithm that
%% chooses vertices (perhaps randomly) and updates each so its value is
%% calculated based on the values of its neighbors.  Now someone notices
%% that most of these updates are independent (only if two neighbors
%% update at the same time do we have a conflict).  Ideally we would like
%% to write a short piece of code that allows for atomic concurrent
%% updates.  We do not want to change the data structure in any way
%% (there are 10s of thousands of lines of code that use it), but we could
%% be willing to annotate loads and stores with an stm version.
%% For example:

%% \begin{verbatim}
%% transaction([&] {
%%   T sum = T::identity;
%%   for (auto y : v.neighbors())
%%   sum = f(sum, stm::load(y.val));
%%   stm::store(v.val, sum);
%% });
%% \end{verbatim}

%% Here we just added a stm::load and stm::store.  The problem is that
%% this will not work with any of the efficient TMs we are aware of, at
%% least not for general types T.  The issue is that prior systems either
%% require that one modify the types of value elements, or are limited to
%% loading and storing trivial types, and only if they fit in 64-bits.
%% In general the type T could be a pair of longs (requiring 128 bits) or
%% even include a vector (which has a destructor).

\section{Algorithm}
\label{sec:algorithm}

Here we describe our algorithm.  We first describe the data structures
we use and then how we implement the various operations.  We present
pseudocode in this section, and the full C++ code is given in the
appendix.  They do not match exactly since we can make some
simplifications in the pseudocode (e.g. the C++ code has to account
for the fact that the memory is not sequentially consistent).

\begin{figure}[t]
\begin{lstlisting}
struct version_link =
  next        // pointer to next link
  data        // value of this version
  stamp       // timestamp of that link
  loc         // pointer to location versioned by this link

global_stamp = 0

thread_local tid
thread_local late_read = false
thread_local start_stamp = load(global_stamp)

next_stamp(prev_stamp) =
  let (curr_stamp, success) =
    CmpX(global_stamp, prev_stamp, prev_stamp + 1)
  if success return prev_stamp + 1
  else return curr_stamp

enum lock_status =
  // Lock acquired by thread with tid
  | Locked tid
  // Lock not acquired, store most recent timestamp
  | Unlocked stamp

struct lock =
  // whether lock currently acquired, and if so by whom
  status : atomic<lock_status>
  // Version list protected by lock
  verlist : atomic<version_link*>

// cons a new link onto the version list protected by a lock
add_link(lck : lock, link : version_link*, prev_stamp) =
  link->next = load(lck.verlist)
  link->stamp = prev_stamp
  store(lck.verlist, link)

enum lock_result = SelfLocked | Acquired stamp | Failed 

try_lock(lck) : lock_result =
  status = load(lck.status)
  case status of
  | Locked tid' =>
    if tid = tid' then return SelfLocked else return Failed
  | Unlocked stamp =>
    if stamp < start_stamp and 
      CAS(lck.status, status, Locked tid)
      return Acquired stamp
    else return Failed

unlock(lck, stamp) = store(lck.status, Unlocked stamp)
\end{lstlisting}
\vspace{-.1in}
\caption{Versioned Lock Pseudocode and API}
\label{fig:alg:lock}
%\vspace{-5pt}
\end{figure}
\subsection{Data structures}

In contrast to most other multiversioned STMs, \ustm{} allows the client to read and write directly to normal memory locations. That is, the most recent version of every value is stored not in a version list, but rather at the location itself. In a sense, the location is the head of the version list, which is detached from the remainder of the list.

Every location is hashed to a lock that protects that location; this function is not necessarily injective, however, and multiple distinct locations may hash to the same lock. Thus, a single lock may protect multiple locations.

Every lock maintains its status---whether locked or unlocked---plus either the thread id of its owner (if locked), or the timestamp of the transaction that most recently updated some location hashing to that lock (if unlocked).

Every lock also maintains a version list of all previous accessible versions of locations hashing to that lock, sorted in non-increasing order of timestamp. Because these version lists are heterogeneous---multiple locations may hash to the same lock---every version link also maintains not only a value and timestamp, but also the location holding that value at the timestamp. The implementation of these version locks is described in Figure~\ref{fig:alg:lock}.
A transaction descriptor maintains information describing the current state of a transaction. It records various information and statistics, including whether the transaction is read-only, a flag to indicate whether some of the reads may have been out of date, whether the transaction is currently in a constructor, the identifier of the thread, and the start timestamp of the transaction.

Additionally, every transaction descriptor maintains read and write logs. Every read log entry just maintains the location read itself, whereas every write log entry also maintains the timestamp of the previous update to that location and a pointer to a (detached) version link that contains the value written to the location within the transaction.

Finally, every transaction also maintains an \emph{allocation} log that records locations allocated within a transaction (which must be deleted to avoid a memory leak if the transaction is aborted). It also maintains a \emph{delete} log comprising all locations deleted within a transaction---these locations are retired upon commit.

\subsection{Loads}

A load on location $l$ first checks whether $l$ is in the write log; if so, it just returns the corresponding value so that the transaction correctly reads its own writes. Otherwise, it adds the location to the read log so that the read can be validated later. Loads then attempt to read out the current value from the location. Doing so naively---by simply reading the bytes from the location---would be unsafe, as a concurrent writer could update the location during the read. Thus, every load must ensure that no concurrent write occurred during the course of this read. It accomplishes this by first recording the timestamp of the most recent update to the lock protecting $l$, reading the bytes out of $l$, and then examining the timestamp again. If the timestamp has not changed and the lock is not acquired, then no write was concurrent with the read, and the bytes read out were valid. Furthermore, if the timestamp is less than the start timestamp of the transaction, the value read corresponds to the most recent value committed before the transaction began, and thus is the correct value to return.

Otherwise, some other transaction may have updated the location since the start timestamp, and the transaction must eventually be aborted if it is not read-only. We record this fact by setting the \texttt{late\_read} flag rather than aborting immediately. This ensures that our system is indeed ``abort-free'' and waits until the client code completes to abort the transaction.

The load then chases down the version list of the lock for location $l$ to find the most recent version link earlier than the start timestamp $t_i$ that matches the location, and returns the associated value.

\subsection{Stores}

Storing value $v$ to a location $l$ first allocates a new, detached
version link that temporarily stores $v$. Then, a log entry containing
this new version link, the location $l$, and value $v$ is appended to
the write log.  This version link is then immediately added to the
delete log for the transaction.  This at first glance seems unsafe,
but the fact that (1) the delete log is not processed until the
transaction completes, and (2) when processed it is retired rather
than deleted, and (3) our integration of epoch-based memory
reclamation within the system ensures that this location will never be
freed while another transaction is still reading it.  This is
implemented by pseudocode in Figure~\ref{fig:alg:store}.

\subsection{Commit}\label{sec:alg:commit}

We now describe the commit phase for a transaction $T_1$ beginning at timestamp $t_i$. If $T_1$'s write log is empty, there is nothing to do besides retiring all of the locations freed by the client code during the speculative phase. Otherwise, if the transaction performed an ``out of date'' read---that is, some transaction $T_2$ following $T_1$ wrote some location read by $T_1$, then $T_1$ is aborted. Otherwise, the transaction tries to acquire every lock protecting a location in the write log. Acquiring a lock \textit{lck} fails if the most recent update to a location protected by \textit{lck} occurred at a time following $t_i$, or if the lock is already taken. If acquiring any of the locks fails, the transaction aborts.

The transaction then reads the timestamp, which becomes the commit timestamp and is assigned
to all writes if the commit succeeds.
The read set of $T_1$ is then validated. To do so, every lock protecting a location in the read log is inspected. If no update following timestamp $t_i$ has written to some location protected by any of these locks, then the validation succeeds. Otherwise, validation fails, and the transaction aborts.
%
%If all of the locks protecting the write log are acquired, and validation of the read log succeed, then the transaction will successfully commit. The end timestamp is fetched from the global clock, which is where the transaction serializes.
%
If all of the locks are acquired, then the transaction will commit. It first fetches the end timestamp $t_f$. With all of the write locks acquired, the entries in the write log are published globally. When a write log entry for location $l$ is applied, the corresponding value $v$ written during the speculative phase is stored in a detached version link. 
%The data stored within the version link, and that currently stored at location $l$ is thus \textit{swapped} so that the version link contains the previous value, and the location contains the new value.

It then \textit{swaps} the data within the version link and location, so that the version link contains the previous value, and the location itself stores the new value written by the transaction. This version link is prepended to the version list for the lock protecting location $l$ (that is currently held by the transaction). Finally, after all writes are applied, the locks held by the transaction are released and the global clock is incremented if it is still equal to the commit stamp. This algorithm is implemented in pseudocode in Figure~\ref{fig:alg:commit}.

If a transaction is aborted, all locks held by the transaction are released, every allocation in the allocation log is freed, and the global clock is incremented if it is equal to the transaction's start timestamp.

\begin{figure}
\begin{lstlisting}
enum last_update = Self | Other stamp

struct write_log_entry =
  loc       // pointer to location written to
  old_stamp // stamp of previous write to location
  link      // version link to be added for previous write
  size      // size in bytes of data stored to location

thread_local tid

thread_local start_stamp

validate_read(lck) : bool =
  case load(lck.status) of
  | Locked tid' => return tid = tid'
  | Unlocked stamp => return stamp < start_stamp

commit() : bool =
  if read_only
    // Never abort
    return true
  if late_read return false
  for every entry `e` in write log
    case try_lock(lock(e.loc)) of
    | Failed => return false
    | Acquired prev_stamp => e.old_stamp <- Other prev_stamp
    | SelfLocked => e.old_stamp <- Self
  commit_stamp = load(global_stamp)
  for every entry `e` in read log
    // validate all reads
    if not validate_read(lock(e.loc)) then return false
  for every entry `e` in write log
    Swap the bytes of `e.loc` and `e.link.data`
    let prev_stamp = case e.old_stamp of
      | Self => commit_stamp
      | Other prev_stamp => prev_stamp
    add_link(lock(e.loc), e.link, prev_stamp)
  Release all acquired locks
  Retire every location in delete log
  // Ensure global stamp is greater than commit stamp
  next_stamp(commit_stamp)

abort_transaction(stm, desc) =
  Release all acquired locks
  Retire all locations in allocation log
  next_stamp(start_stamp)

\end{lstlisting}
\caption{Transaction Commit and Abort}
\label{fig:alg:commit}
\end{figure}

\subsection{Transactions}\label{sec:alg:txn}

We now describe the process for running a thunk $f$ containing client code within a transaction $T$. First, the start timestamp $t_i$ is fetched. Then, the thunk $f$ is executed. Recall that a \texttt{late\_read} flag is set during transaction execution if the transaction performed an out-of-date read---in particular a transaction serializing after $t_i$ committed a data item that was read by the transaction. If so, and furthermore $T$ is not a read-only transaction, then $T$ is aborted and retried. Otherwise, the transaction attempts to commit according to the logic in Section~\ref{sec:alg:commit}. If this process succeeds, then the transaction has committed its writes (if any). Otherwise, the transaction is aborted and retried.

\subsection{Memory Management}\label{sec:alg:mm}

We employ epoch-based memory reclamation (EBR)~\cite{fraser2004practical} to manage shared pointers. Every thread participating in EBR \textit{announces} when it enters a critical section, and \textit{unannounces} when it exits the critical section. A global epoch approximates real time, and is incremented whenever every thread has announced the current epoch. When a thread retires a memory location, this location is placed into a \textit{limbo list} for the current epoch. A limbo list maintaining retired pointers from the previous epoch is also maintained. When the global epoch is incremented, all of the locations in the oldest limbo list are freed, and the current limbo list becomes that for the previous epoch (which was just incremented). The limbo list for the previous epoch is reset to empty. We employ a custom implementation of EBR \texttt{uepoch} that uses thread-local limbo lists and only frees pointers from epoch at most $e - 3$ where $e$ is the current global epoch. Inside a transaction, an epoch is announced before taking a start timestamp and then unannounced after running client code.

\begin{figure}[t]
\begin{lstlisting}
load(loc) =
  if loc is in write log
    return data stored in associated version link
  let lck = lock(loc)
  Let buf be a temporary buffer
  hd <- nullptr // Pointer to head of version list for lck
  // Timestamp of most recent update to lck
  last_stamp <- -1
  while true do
    initial_status = load(lck.status)
    case initial_status of
    | Locked _ =>
      // version lock acquired, try again
      continue
    | Unlocked stamp =>
      Copy the bytes of loc to buf
      hd <- load(lck.verlist)
      let final_status = load(lck.status)
      if initial_status != final_status
        // lock state changed during read
        continue
      last_stamp <- stamp
      if last_stamp < start_stamp
        // last update serialized before txn start
        return buf
      // Otherwise read was consistent but out of date
      break
  done
  // Another transaction committed a write after we began
  late_read <- true
  // Chase down version list to find correct version
  result <- buf
  while last_stamp >= start_stamp and hd != null
    if hd->loc = loc then result <- hd->data
    last_stamp <- hd->stamp
    hd <- hd->next
  // Deepest matching link holds the correct version
  return result

store(loc, val) =
  // Allocate version link storing new value to write
  let link = new version_link {
    next = nullptr,
    data = val,
    stamp = None,
    loc = loc
  }
  // Alloc-swap-retire: retired when txn completes
  Add link to delete log
  let entry = {
    loc = loc,
    data = val,
    old_stamp = None,
    link = link,
  }
  Add `entry` to write log
\end{lstlisting}
\caption{Transactional Load and Store Pseudocode}
\label{fig:alg:load}
\label{fig:alg:store}
% \end{minipage}
\end{figure}

\subsection{Synchronization Between Memory Management and the Global Clock}

Recall that when a transactional write of value $v$ to location $l$ is performed, a tentative version link with value $v$ is created and then immediately placed in the delete log. When (if) the transaction is committed, every item in the delete log is retired, including version links for previous values of locations that were written to by the transaction.

We must ensure that no transaction attempts to read a version link that has already been freed. In particular, consider a transaction $T$ beginning at $t_i$ that scans down the version list for location $l$. It searches for the most recent version link with timestamp less than $t_i$. The worry is that this version link could be retired, which is indeed possible with split timestamps. In particular, another transaction $T'$ committing a data item read by $T$ could serialize at $t_i$, in which case $T$ continues scanning down the version list past that committed by $T'$, which may be garbage. We resolve this by incrementing the timestamp \textit{before} incrementing the epoch. Intuitively, this ensures memory safety by maintaining the invariant that the global clock is always at least the epoch so that transactions do not attempt to read too far into the past. Formally, we have the following theorem:

\begin{theorem}[Memory Safety] No transaction accesses a freed link.
\label{theorem:memsafe}
\end{theorem}

\begin{proof}[Proof (Sketch)]

  Consider a transaction $T$ with start stamp $t_i$ traversing a version list and any link $\ell$ this traversal reaches. Let $c$ be the commit stamp of the transaction $T'$ that committed $\ell$, and let $e$ be the value of the global epoch when $c$ was read from the global clock. Because $\ell$ is traversed, it must be the case that $t_i \leq c$, as the stamp associated with $\ell$ is at most $c$. Every increment of the epoch past $c$ (beyond at most one that may have been in-flight prior to the read of $c$ by $T'$) is preceded by an increment of the clock. $T$ read start stamp $t_i \leq c$ after announcing its epoch, so $T$ announces epoch at most $e + 1$. The epoch cannot advance beyond epoch $e + 2$ during the execution of $T$, and \texttt{uepoch} only frees pointers from epochs that are older than 2 less than the global epoch, so $T$ never traverses a freed link. The full proof is given in the full version~\cite{fullpaper}.

\end{proof}

\section{Correctness}\label{sec:correctness}

In this section we outline a proof of correctness of the approach. Various different correctness criteria exist for transactional systems, including (strict) serializability, and opacity. Serializability requires that all transactions appear to take place \textit{atomically} in some serialized order. Strict serializability furthermore requires that this serialized order preserves the real-time order of transactions---i.e. if $T_a$ commits before the invocation of $T_b$, then $T_a$ precedes $T_b$ in the serialized order.

Strict serializability is typically given as the strongest correctness criterion within the database community, but is too weak for STM systems. (Strict) serializability speaks only of \textit{committed} transactions, whereas the semantics of STM systems must also consider the behavior of aborted transactions. In particular it is desirable that even aborted transactions see only a ``consistent'' snapshot of shared state. Otherwise, programmers may make assumptions that do not hold inside of aborting transactions. \textit{Opacity} guarantees exactly this~\cite{opacity08}, requiring that even aborted transactions serialize at some point between invocation and abort. 

As stated by the following theorem, \ustm{} guarantees opacity.

\begin{theorem}[Opacity]
  Any history of \ustm{} transactions is opaque.
\end{theorem}

\begin{proof}[Proof (Sketch)]

  It suffices to show that all transactions serialize at some point between their invocation and response, including aborted transactions. Read-only transactions serialize when they read the start stamp. Aborted (update) transactions also serialize when they read the start stamp.

  A committed update transaction with commit stamp $t_f$ serializes when the global clock is incremented from $t_f$ to $t_f + 1$, as this is when the writes become globally visible to readers. Note that this may be \textit{after} the locks are released by the updating transaction, as the update transaction invokes \texttt{next\_stamp} after the locks are released. Regardless, this increment must occur before the transaction returns.

  Intuitively, these serialization points are consistent because a read by a transaction with start stamp $t_i$ will only observe the writes performed by an update transaction with commit stamp $t_f < t_i$. Hence when the global clock advances beyond $t_f$, all writes installed by that update become globally visible. Note that multiple transactions may share the same commit stamp, and thus one clock increment may serialize multiple updaters. The write sets of such transactions must be disjoint, as otherwise lock acquisition would fail. There can still exist anti-dependencies between transactions sharing the same timestamp. Consider transactions $T_1$ and $T_2$ that read and write location $x$, respectively. If $T_1$ validates the lock protecting $x$ before $T_2$ acquires it, then this orders $T_1$ before $T_2$ yet they share the same commit stamp. However, these (anti) dependencies are acyclic, and such transactions may be serialized topologically. The full proof is given in the full version~\cite{fullpaper}.

\end{proof}

\section{Experiments}\label{sec:experiments}

We evaluate \ustm{} on a variety of different workloads, comparing its performance to other state of the art STMs, demonstrating that \ustm{} matches or exceeds performance of these systems without sacrificing simplicity or generality.

\myparagraph{Setup} All experiments are run on a 96-core Amazon Web Services c7i-metal instance with 2x Intel(R) Xeon(R) Platinum 8488C (48 cores and 3.2 GHz), and 384 GB memory. Each core is 2-way hyperthreaded, giving 192 hyperthreads. 
%All benchmarks were run using \texttt{numactl -i all}, spreading memory pages across sockets in a round-robin fashion.
The machine runs with Ubuntu 22.04.1 LTS, and the code was compiled using g++11 with -O3.

\myparagraph{Systems Tested} We benchmark against \fuse{}~\cite{tlf25}, Multiverse~\cite{Coccimiglio0R26}, and \twoplsf{}~\cite{2plsf}. The first two are state of the art optimistic multiversioned TMs, whereas \twoplsf{} is a state of the art pessimistic single-versioned TM, allowing for a robust comparison across different design spaces.
We do not present experiments for DCTL~\cite{RC24}, as its implementation is proprietary, or TinySTM~\cite{tiny08}, because it could not execute without crashing in most experiments.

% Under high contention, the throughput of Multiverse is limited by the heartbeat of the central clock. Although they employ the same lazy timestamp algorithm of~\cite{RC24}, they use a \cas{} to increment the global clock if it is equal to the start timestamp on transaction abort. Even if the \cas{} fails, it still takes the cache line in exclusive mode. Because the clock is monotonically increasing, it is correct and more efficient to first load the value of the global clock and only attempt the \cas{} if the value loaded is indeed equal to the start timestamp. We implemented this change, resulting in significant improvements to Multiverse's performance under moderate to high contention. The authors also integrated this change when we notified them. However, this modification results in a segmentation fault on the c7i-metal instance on which we ran our final benchmark suite, and we were forced to exclude it for the final timings.

\myparagraph{Workloads} Our workloads are based on those from YCSB benchmark suite~\cite{YCSB}, which is commonly used to benchmark key-value stores\footnote{Our workloads are not \emph{literally} taken from YCSB, since we do not classify different parameter regimes into workloads A/B/C/D/E as in YCSB}. Our benchmarks consist of measuring the throughput of transactions consisting of inserts, finds, and deletes to random keys within a key-value store. We implement this store using various data structures and vary different YCSB parameters, measuring the throughput over this mix. In particular, we implement the key-value store (in different experiments) using a Linked List, Skip List, B-tree, Adaptive Radix Tree (ART)~\cite{art13}, Treap, AVL tree, Leaf Tree, and Hash Table. The leaf tree is a simple binary tree where data is only stored at the leaves. All of these data structures are simply sequential implementations. 

% It is likely possible to achieve higher throughput using state-of-the-art concurrent data structures, but we wish to demonstrate that \ustm{} 

For each backing data structure, we vary (a) data structure size (denoted by $n$), (b) update percentage, (c) number of operations per transaction, (d) number of threads, and (e) zipfian parameter. Every data structure is initially prefilled to size $n$ with keys selected uniformly at random from a universe $U$ of $2n$ 64-bit keys total. Every key is associated with a corresponding 64-bit value.

In the timed portion of the code, each thread executes transactions consisting of inserts and deletes (in equal numbers) and finds. Keys for these operations are also sampled from $U$, but according to a zipfian distribution specified by $z$. $z$ ranges from 0 (uniform) to 0.99 (highly skewed).   This models common access patterns to databases, in which most accesses are concentrated around ``hot'' keys. We measure the throughput in operations per second over different mixes of parameters and backing data structures.

When unspecified, we fix every parameter at its \textit{default}. The default size for list data structures is 300, whereas for all others it is $10^6$. These remaining default values are $u = 5\%$ updates, $t = 4$ operations per transaction, $p = 192$ threads, and zipfian $z = 0.75$.

\subsection{Geometric Mean Performance}\label{sec:experiments:geomean}

Our first set of experiments aim to compare the average performance of the STM systems under test over a wide variety of different workloads. To this end, we measure the throughput of every system for every combination of the following parameters:

\begin{itemize}
  \item Update Percentage $\in \{ 0\%, 5\%, 50\%\}$
  \item $n \in \{ 10^2, 10^3 \}$ for linked list, $n \in \{ 10^5, 10^6, 10^7 \}$ for all other structures
  \item Transaction Size $\in \{ 1, 4, 16 \}$
  \item Zipfian $\in \{ 0, 0.75, 0.99 \}$.
\end{itemize}

For every STM system and backing data structure, we then calculate the geometric mean of the throughputs across this mix of parameters. The results are shown in Figure~\ref{fig:geomean}, where the geometric means are grouped by data structure and normalized to the max per structure. As we can see, \ustm{} achieves the highest throughput across all data structures over this parameter mix.

\subsection{Varying Parameters}

Our next set of experiments fixes all but one YCSB parameter, which is varied across a range of different values. We then measure how the performance of each system changes as this parameter changes, allowing us to compare the relative performance of different STMs on different workloads. We evaluate one data structure with high fanout (the B-tree), another with low fanout (the AVL tree), and the hashtable; scaling within each class is similar.

\begin{figure}[t]
\centering
\includegraphics[width=\columnwidth]{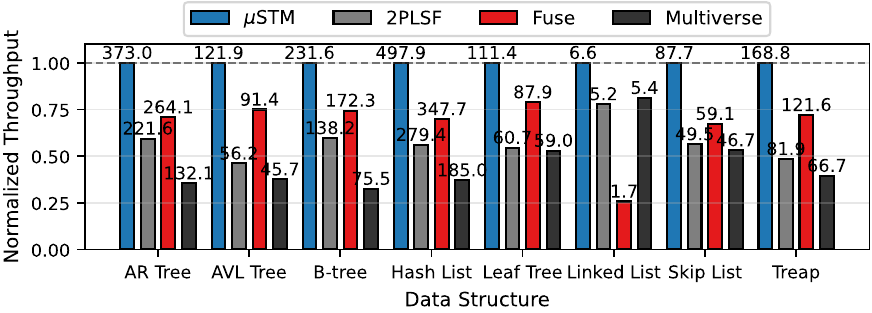}
\caption{Geometric mean of throughput across 3 sizes, 3
zipfian parameters, 3 update\%, and 3 transaction
sizes. Normalized to the highest throughput per structure.}
\label{fig:geomean}
\end{figure}

\begin{figure}[t]
\begin{subfigure}[t]{0.48\textwidth}%
\centering
\includegraphics[width=\columnwidth]{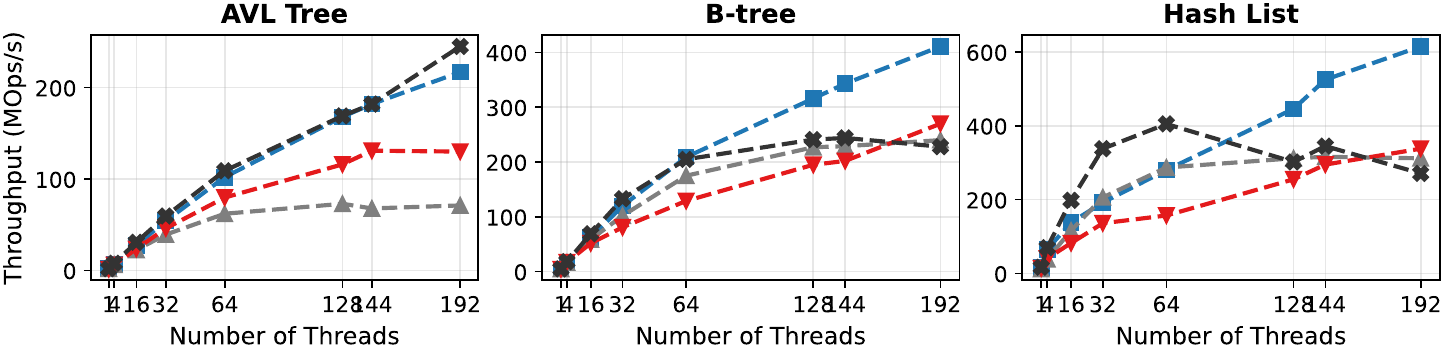}
\caption{Varying Thread Count}
\label{fig:threads}
\end{subfigure}
\hfill
\begin{subfigure}[t]{0.48\textwidth}
\centering
\includegraphics[width=\columnwidth]{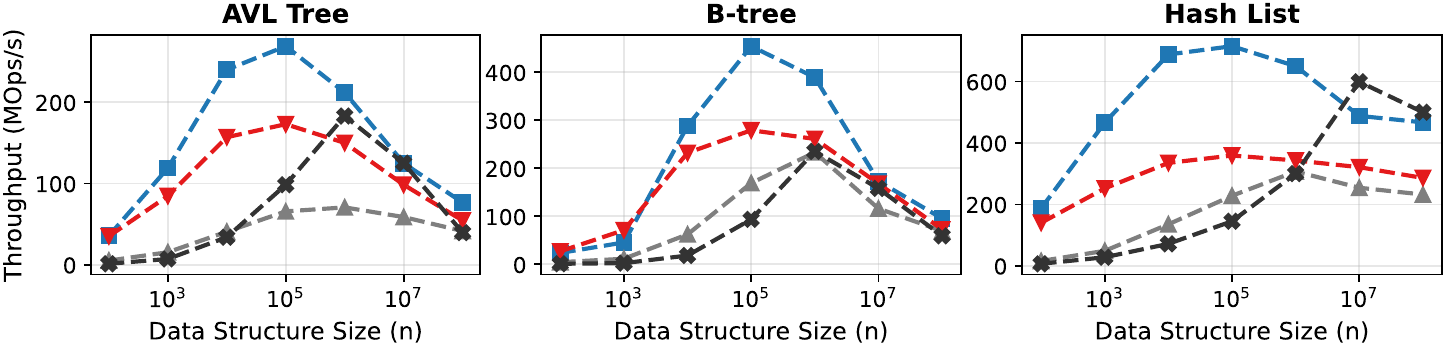}
\caption{Varying data structure size}
\label{fig:size}
\end{subfigure}
\hfill
\begin{subfigure}[t]{0.48\textwidth}
\centering
\includegraphics[width=\columnwidth]{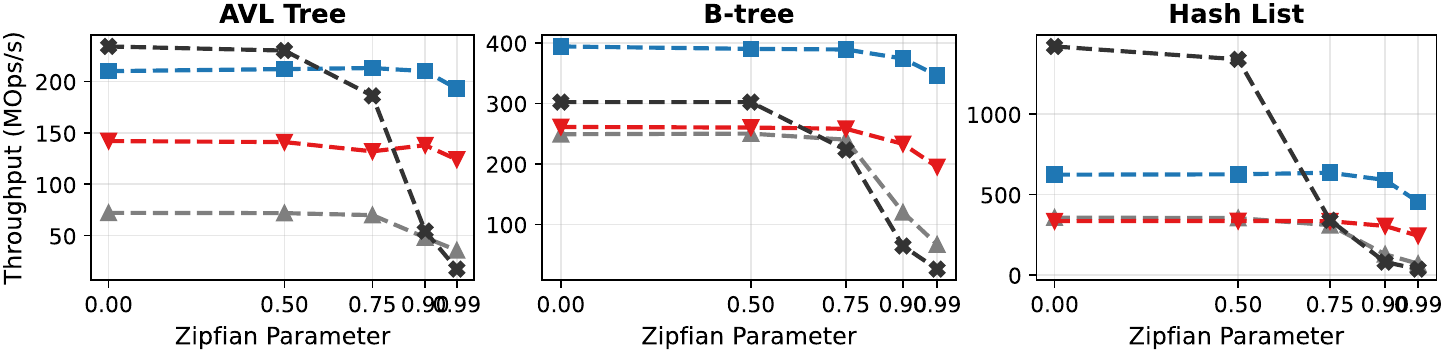}
\caption{Varying skew}
\label{fig:zipf}
\end{subfigure}
\hfill
\begin{subfigure}[t]{0.48\textwidth}
\centering
\includegraphics[width=\columnwidth]{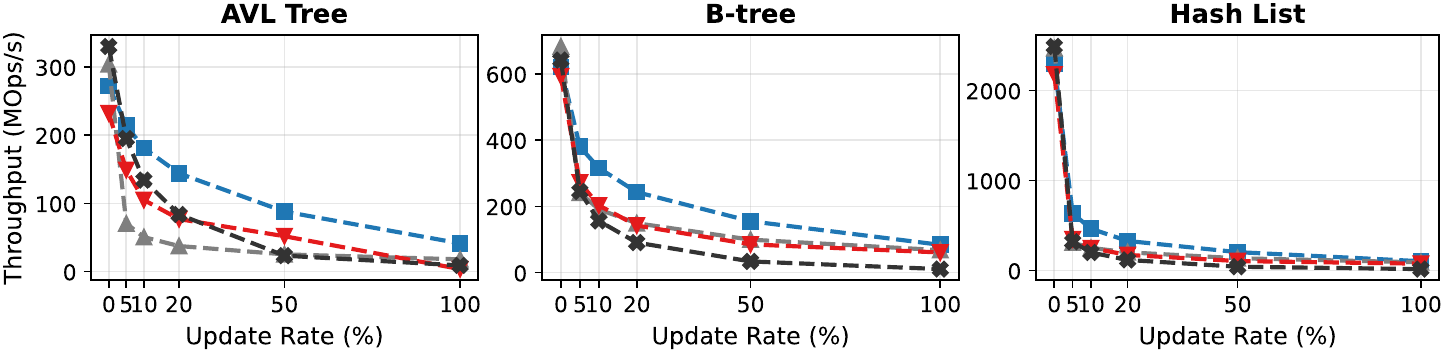}
\caption{Varying update rate}
\label{fig:update}
\end{subfigure}
\hfill
\begin{subfigure}[t]{0.48\textwidth}
\centering
\includegraphics[width=\columnwidth]{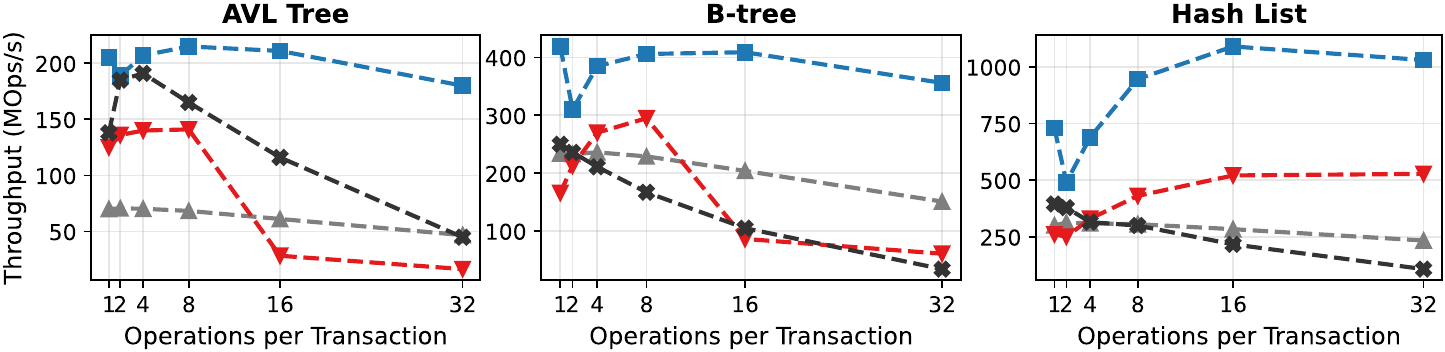}
\caption{Varying transaction size}
\label{fig:txn-size}
\end{subfigure}
\hfill
\caption{
Comparison of throughput between \twoplsf{}, \fuse{}, Multiverse, and \ustm{} (higher is better).
}
\end{figure}%
\myparagraph{Thread Count} Figure~\ref{fig:threads} displays the scalability of the different STM systems with respect to thread count. We see that all STM systems scale well with thread count except \twoplsf{}, which levels off around 128 threads. This is because their implementation of reader-writer locks is not scalable. In particular, for $p$ threads, acquiring a write lock requires scanning $p$ read indicators to ensure that no thread has taken a read lock.

% \begin{figure}[ht]
%     \centering
%     \includegraphics[width=\columnwidth]{plots/paper/threads.pdf}
%     % \vspace{-.15in}
%     \caption{Varying Thread Count}
%     \label{fig:threads}
% \end{figure}

Multiverse also does not scale well beyond 128 threads on data structures with high fanout like the B-tree. Because these trees have wide fanout and are thus shallow, most writes are concentrated on a few select nodes along the root to leaf path. Furthermore, Multiverse employs \textit{eager} locking for writes, acquiring write locks during the traversal phase. Readers will abort if they encounter a lock acquired by a writer \textit{even} if that writer will later abort. Hence writers can starve readers even if they later abort, and this is more likely to occur in data structures with high fanout.

% \begin{figure}[ht]
%     \centering
%     \includegraphics[width=\columnwidth]{plots/paper/size.pdf}
%     % \vspace{-.15in}
%     \caption{Varying data structure size}
%     \label{fig:size}
% \end{figure}

\myparagraph{Data structure size} We examine how the throughput of different STM systems changes with data structure size in Figure~\ref{fig:size}. We see that \ustm{} scales well up to $\sim 10^6$. At this point, the data structure likely no longer fits in L2 cache, and beyond $10^7$ the working set can no longer fit in the L3 cache. \ustm{} is designed so that the lock table fits in L3 cache, so this is expected beyond this point. The program becomes memory bound, and the throughput of all STMs degrades as expected.

\myparagraph{Zipfian} Figure~\ref{fig:zipf} compares the throughput of different systems while increasing the skew of the key distribution. When keys are uniformly distributed, all STM systems achieve high throughput, with Multiverse outperforming all other systems on some data structures, like in the Hash Table.

% \begin{figure}[ht]
%     \centering
%     \includegraphics[width=\columnwidth]{plots/paper/zipfian.pdf}
%     % \vspace{-.15in}
%     \caption{Varying skew}
%     \label{fig:zipf}
% \end{figure}

The throughput of \ustm{} and \fuse{} is mostly stable as zipfian increases, whereas the throughput of Multiverse and \twoplsf{} falls significantly even at the relatively low default update rate (5\%). Again, this is especially pronounced for data structures with high fanout where the average traversal is short. We believe that this is due to \twoplsf{} and Multiverse's eager acquisition of locks.

Additionally, we see that Multiverse performs extremely well for data structures with a short traversal at low contention---this is especially noticeable for the Hash Table, where Multiverse achieves twice the throughput of \ustm{} at low zipfian. This is one regime in which lazy timestamping performs much better than split timestamping. In \ustm{}, an update transaction must increment the global clock if it has not changed between when the commit stamp is taken and when the locks are released. This is more likely for data structures with a short traversal, as transactions over these data structures will have a small read log to validate and write log to apply. In contrast, in the lazy timestamp algorithm employed by Multiverse, the global clock is only incremented on abort, which at low zipfians is highly infrequent. Furthermore, in hash tables there is a lower likelihood of two transactions conflicting due to hash buckets being independent components. Thus in this regime the heartbeat of the global clock is a bottleneck for \ustm{} but not Multiverse.

\myparagraph{Update Rate} Figure~\ref{fig:update} compares the throughput of different systems for increasing update rates. For read-only transactions, Multiverse often achieves the highest throughput of any STM. Because locations are only versioned when contended, Multiverse operates in single-versioned mode for the read-only workload, achieving high throughput. \twoplsf{} also performs well for read-only workloads; the implementation of scalable read indicators distributes the read indicators for different threads across uncontended cache lines, minimizing overhead for read-only transactions.

However, the throughput of \twoplsf{} and Multiverse quickly declines as update rate increases. This is because, as discussed, both \twoplsf{} and Multiverse suffer at high contention. Furthermore, \twoplsf{}'s implementation of scalable reader-writer locks penalizes writers by forcing every writer to scan $p$ read indicators to acquire a write lock. Again, the decline in Multiverse's throughput is not uniform across data structures, and is more pronounced for data structures with higher fanout like the B-tree.

\myparagraph{Transaction Size} In Figure~\ref{fig:txn-size}, we compare the throughput of different implementations against varying transaction size. We see that the performance of \twoplsf{} and \ustm{} decays only modestly with an increasing number of operations per transaction, whereas that of \fuse{} and Multiverse quickly drops off. Again this is because Multiverse suffers under contention, and a larger number of operations per transaction increases conflicts.

Additionally, we see that \ustm{} performs relatively poorly at a small number of operations per transaction for the chaining hash table, but improves markedly with a larger number of operations per transaction. This is because lazy-timestamping generally performs better than split-timestamping under low contention with short-running transactions. Under such workloads---especially those with few operations per transaction---the commit phase of every transaction is very short. For \ustm{}, this means that it is less likely that the global clock was incremented between the point when the commit stamp of a transaction is read and when it is later possibly incremented after the locks are released, and more transactions have to increment the global clock.

% \begin{figure}[ht]
%     \centering
%     \includegraphics[width=\columnwidth]{plots/paper/ops_per_txn.pdf}
%     % \vspace{-.15in}
%     \caption{Varying transaction size}
%     \label{fig:txn-size}
% \end{figure}

\subsection{Range Queries}

We compare the throughput of different STMs for \textit{range queries}. For this experiment, 50\% of threads (the writers) execute transactions consisting of four update operations at keys uniformly sampled from $U$ ($z = 0$). The remaining 50\% of threads execute \textit{range queries}, which are read-only transactions that uniformly sample a start key from $U$ and perform contiguous read operations starting from that key for a specified \textit{range size}. A range query samples a start key $k$ and reads keys [$k, k$ + $ range\_size$]. %+ does not appear inside math mode for some reason

\begin{figure}[ht]
    \centering
    \includegraphics[width=\columnwidth]{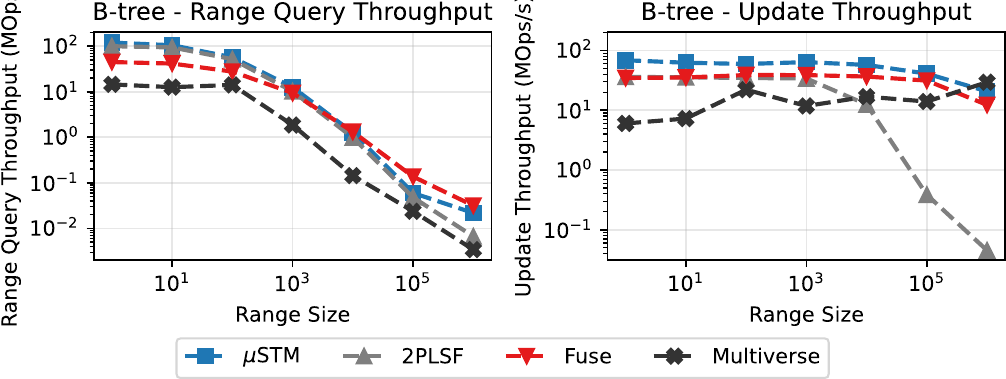}
    % \vspace{-.15in}
    \caption{Range Query Throughput. Left is the throughput of the range queries themselves, right is the throughput of updaters across increasing range sizes}
    \label{fig:rq}
\end{figure}

In Figure~\ref{fig:rq}, we examine how the throughput of both the range queries and update transactions change with increasing range size on the B-tree---the data structure supporting most efficient range scans. \fuse{} and \ustm{} generally achieve the highest range query throughput across all sizes. We see that the range query throughput scales inversely with range query size, as expected. However, the range query throughput of \twoplsf{} does not scale significantly worse than the other multiversioned STMs, which may be surprising. In contrast, the update throughput of \twoplsf{} falls greatly with higher range size, whereas that of other STMs is stable. This is because of how \twoplsf{} arbitrates conflicts between different transactions to ensure starvation-freedom. Transactions with earlier start timestamps are given priority, and can abort those with later timestamps. Long-running range queries will generally have lower timestamps than newer update transactions, giving the range queries priority and aborting the short-lived writers.

\subsection{Ablation Studies}\label{sec:ablation}

We now perform two ablation studies to determine how removing different features of \fuse{} and \ustm{} affects throughput. The results are shown in Figure~\ref{fig:ablation}.

\begin{figure}[ht]
    \centering
    \includegraphics[width=\columnwidth]{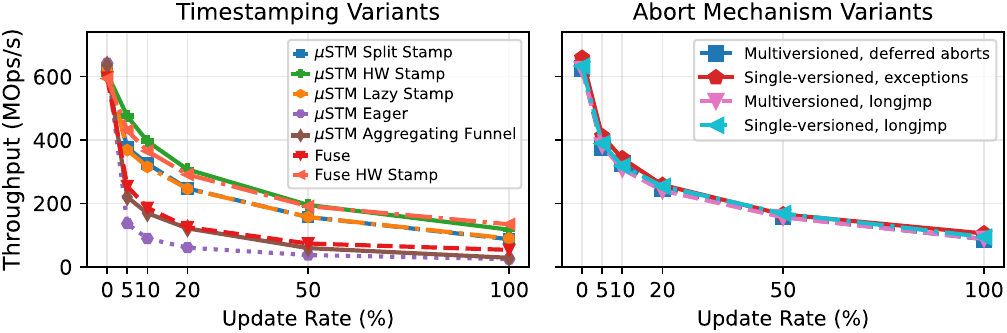}
    % \vspace{-.15in}
    \caption{Ablation studies}
    \label{fig:ablation}
\end{figure}

\myparagraph{Timestamp Algorithms} In one ablation study we compare the performance of variants of \fuse{} and \ustm{} that employ different timestamp algorithms. We measure the throughput of these variants on the same YCSB-like benchmark for the B-tree across increasing update rate. The two variants \texttt{\fuse{}-HWStamp} and \texttt{\ustm{}-HWStamp} use a hardware counter (based on the x86 \texttt{rdtsc} instruction). The default variant of \fuse{} uses an eager timestamp mechanism, incrementing the clock on every transaction. We also include a variant of \ustm{} that employs the lazy timestamp algorithm introduced by~\cite{RC24}, and two \textit{eager} variants that increment the clock on every transaction. One simply uses a hardware fetch and add, while the other employs a more complex software implementation of fetch-and-add, aggregating funnels~\cite{aggregatingfunnels}. Aggregating funnels use software combining to batch different fetch-and-add operations.

We see that both hardware timestamp algorithms are the fastest across all update rates, as expected. The split and lazy timestamp algorithms achieve basically equal performance, and are only $\sim$20\% slower than the hardware variants. Moreover, they scale well with update rate---the drop in throughput is not greater than the hardware timestamp variants. Hence the heartbeat of the central software clock is not a bottleneck, even at high contention.

Finally, we see that all eager variants, including the default timestamp variant of \fuse{} and both eager variants of \ustm{}, achieve relatively poor performance with increasing update rate. The vanilla eager variant of \ustm{} achieves especially poor performance, with performance dropping dramatically even at 5\% updates. The variant based on aggregating funnels achieves throughput that is generally twice as high across increasing update rate, but is still comparably low. Thus, the bottleneck in STMs that use eager timestamp algorithms is the heartbeat of the central clock.

\myparagraph{Abort variants} Our second ablation study compares the throughput of variants of \ustm{} employing different abort strategies across increasing update rate. We include the default multiversioned abort-free variant, a multiversioned variant with early aborts implemented using \texttt{longjmp}, and two single-versioned variants which by necessity must abort early. One variant also uses \texttt{longjmp} to implement aborts, while the other uses exceptions. Overall, the throughput of all versions is comparable. The single-versioned implementations achieve slightly higher throughput for read-only workloads, but fall off with higher update rate. For all update rates, we see that the overhead of multiversioning is relatively low. Moreover, early aborts do not achieve higher throughput by avoiding wasted work.

\subsection{Comparison to Fine-Grained Concurrency}

In this section we demonstrate that \ustm{} does not introduce significant overhead relative to fine-grained concurrent data structures---in particular, a B-tree, AVL tree, and hashtable implemented using \emph{optimistic locking} (OL)~\cite{KungL80}. OL is a technique that allows most of the traversal in these search structures to proceed without locks.

Our baseline measures the throughput of singular OL operations on the same YCSB workload \emph{not} wrapped in any transaction. We compare this to the throughput of \emph{transactions} over the corresponding STM data structures comprising $t$ operations.

% Note here we are running transactions over the OL data structures, whereas in the previous experiment the transactions were over STM data structures. This experiment would not have been possible to implement using sequential data structures, as it is not safe to run sequential code concurrently outside of a transaction.
% running $t$ operations over OL data structures in non-serializable groups; each operation within a group just runs one after another with no transactional bookkeeping or guarantees.

For this experiment, we set $n = 10^6$ and $z = 0$ to measure transactional overhead and reduce the confounding effects introduced by contention-induced aborts; the results over increasing $t$ are shown in Figure~\ref{fig:txn-overhead}. We see that the throughput of the baseline is relatively stable across increasing $t$, whereas for the various STM systems it increases up to $t = 8$ and then stabilizes; at this point, the cost of the operations themselves dominates the startup cost. Overall, the disparity between \ustm{} and the baseline is relatively low, stabilizing at about 20\% for the B-tree and the Hashtable.

\begin{figure}[ht]
    \centering
    \includegraphics[width=\columnwidth]{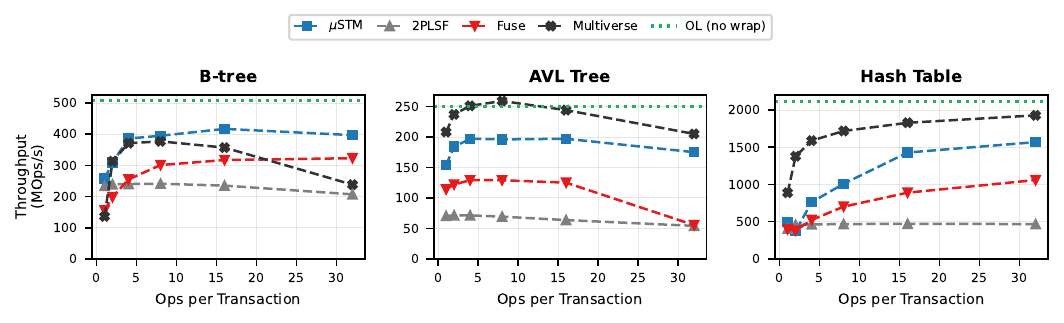}
    % \vspace{-.15in}
    \caption{Evaluation of the transactional overhead of different STM systems over different numbers of operations per transaction on different Optimistic Locking (OL) data structures. The green line represents the throughput of OL operations not wrapped in transactions.}
    \label{fig:txn-overhead}
\end{figure}

\subsection{Non-trivial Types}

We now evaluate the performance of \ustm{} over the same YCSB-like benchmark when used with non-trivial types. The key-value store for this experiment is backed by a probing hashtable that stores buckets inline. We compare the throughput of this hashtable with three types of keys and values: integers, \emph{short strings}, and \emph{long strings}. The strings are implemented using a \texttt{parlay::sequence} that supports the same API as \texttt{std::vector}. The container uses a short string optimization (SSO) that stores the string itself inline with the container, whereas the long string must be stored through a level of indirection. In both cases the size of the \texttt{parlay::sequence} is 16 bytes (but with the long string stored via an additional indirect slot). The integer key is 8 bytes. For each class of value, the value itself is fixed at an arbitrary value: a 1-byte string for short string and 20-byte string for long string. For short strings, the key space is the set of string representations of every integer in the range $[0, 2n)$, whereas for long strings it is the string representation of the hash of every integer in this range. For $n = 10^6$, every short string key fits inline, whereas it is highly likely that every long string value does not.

\begin{figure}[t]
    \begin{minipage}[t]{0.48\columnwidth}
        \centering
        \includegraphics[width=\linewidth]{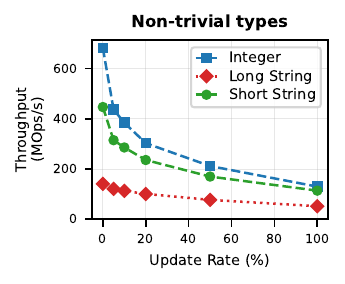}
        \caption{Throughput of a probing hash table on a YCSB-like benchmark with different key/value sizes.}
        \label{fig:strings}
    \end{minipage}
    \hfill
    \begin{minipage}[t]{0.48\columnwidth}
        \centering
        \includegraphics[width=\linewidth]{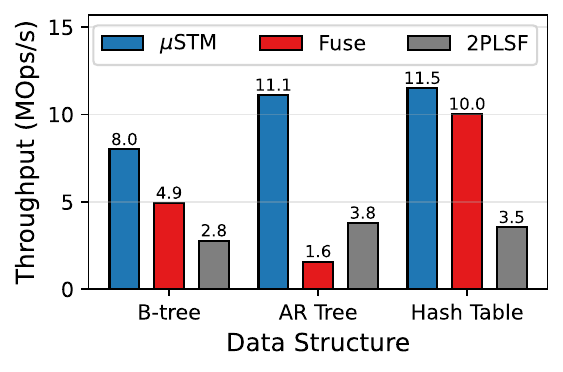}
        \caption{TPC-C Benchmark}
        \label{fig:tpcc}
    \end{minipage}
\end{figure}

Results are shown in Figure~\ref{fig:strings} across increasing update rate where the other YCSB-like parameters are fixed at their defaults. We see, as expected, that the variant employing integer keys achieves the highest throughput, with the variant employing short strings achieving about 20\% lower throughput. This disparity is due to several factors. Each bucket consists of one 8-byte integer key plus the size of the value; thus each bucket employing integer values requires 16 bytes, whereas that employing short strings requires 32. Hence the data structure is twice as large, and more of it resides in L3 vs.\ L2 cache. Additionally, the cost of transactional reads to larger data types is higher, especially under contention; the sequence lock mechanism of reads requires reading the value until it is stable.

The disparity in throughput between short and long strings is much greater; every insert of a long string entry requires two allocations (one for the key, and one for the value), and hashing/probing this key requires a level of indirection to access the string itself.

\subsection{TPC-C Benchmarks}

%\begin{figure}[ht]

We evaluate the three highest-performing transactional data structures on a TPC-C-like benchmark suite~\cite{tpcc}. TPC-C is commonly used to benchmark Online Transactional Processing (OLTP) systems. TPC-C tables maintain data for different \textit{warehouses}, like the stock of items available there, and transactions mutate these data. For our experiments, the number of warehouses is fixed at 192 (the number of hardware threads). The results of our TPC-C benchmark are shown in Figure~\ref{fig:tpcc}, which measures the throughput of databases backed by an ART, B-tree, and hash table. \ustm{} achieves the highest throughput across all data structures---this is particularly notable for the ART.

\subsection{Comparing Architectures}

Finally, we benchmark each system on two additional architectures.  These
results include the Intel machine used in previous experiments, an 80-core ARM
Neoverse N1 (up to 3GHz, 1 NUMA socket, no hyperthreading, and no L3 cache) and
a 96-core AMD EPYC 9R14 (up to 3.3GHz, 1 NUMA socket, 256 MB L3 Cache, 2-way
hyperthreading).  These machines present a diverse architectural spectrum as
the AMD and Intel machines implement the same ISA but with different
microarchitectural choices, and the ARM machine has a different ISA and design
philosophy (e.g., relaxed memory ordering and no L3 cache).

Results for each machine on AVL Tree, B-Tree, and Hash Table structures can be
seen in Figure \ref{fig:machines}.  On both Intel and AMD \ustm{} completely
outmatches the other STMs, performing $33\%\!-\!47.5\%$ better when compared to
the next best system in each benchmark. Surprisingly, while on Intel there is a
clear hierarchy, on AMD the story changes and \fuse{}, \twoplsf{} and
\multiverse{} have relatively equivalent performance.  Lastly, on ARM \fuse{}
becomes a close competitor on all data structures and \multiverse{} performs
just as well as \ustm{} on the hash table.
We note that \ustm{} had significant performance degradation for read-heavy
workloads on ARM, due to designing the lock table to fit in L3 cache (which the
ARM machine does not possess). Indeed in our experiments, making the lock table
smaller improved \ustm{} performance in these cases. Still, we decided not to
show results with this modification in order to refrain from hyper-optimizing
to a certain hardware.

\begin{figure}[t]
\centering
\includegraphics[width=\columnwidth]{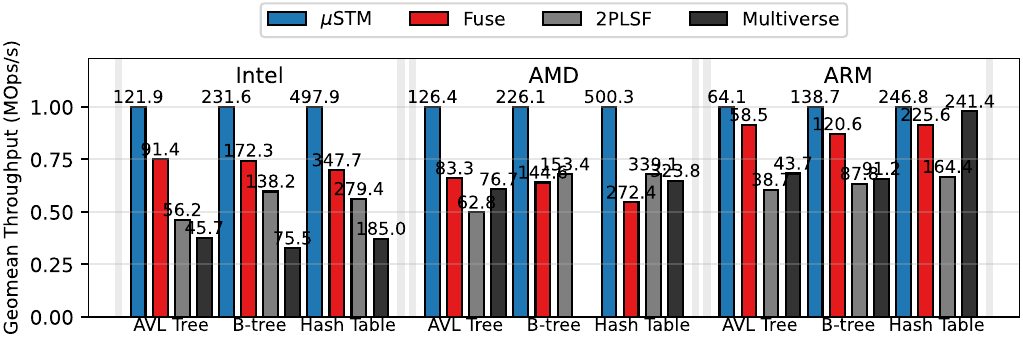}
%\hfill
\caption{Comparing AVL Tree, B-Tree and Hash Table on Intel, AMD and ARM.}
\label{fig:machines}
\end{figure}

\section{Discussion}\label{sec:discussion}

We discuss some of the implementation details and limitations that guided our
design of~\ustm{}.

\subsection{Hardware Timestamps}\label{sec:hwstamps}

Timestamping is most commonly implemented as a single shared counter that is
incremented atomically. Yet, as we show in Section~\ref{sec:ablation}, a more
performant variant involves utilizing hardware cycle counters to establish a
happens-before relation.  However, a cycle counter is required to fulfill two
key properties to qualify~\cite{Ruan13}: 1) processors see their own clock as
strictly monotonic (locally monotonic) and 2) if two instructions executed
concurrently are ordered, then their clock values must reflect the same
ordering (globally monotonic).

The problem lies in ensuring that the hardware actually provides these
properties. Ruan et al.~\cite{Ruan13} cite private conversations with an Intel
engineer regarding the guarantees provided by the \texttt{rdtscp} instruction.
Although on Intel we have indeed experimentally observed these properties, the
same cannot be said for AMD machines, even though they implement the same ISA.
Similarly, Kashyap et al.~\cite{ordo} cite private conversations when mentioning
that clocks in Intel machines have constant skews, an assumption they heavily
rely on. On ARM machines, we were able to successfully use hardware cycle
counters (based on Linux Kernel's implementation~\cite{linux_arch_timer}) and
observed similar results to the ones in Section~\ref{sec:ablation}.  Since
these properties remain undocumented, we view hardware stamping as a
non-portable alternative to lazy or split timestamping that should be used when
available and appropriate.

\subsection{Sequence Locking}
To load and store non-atomic user data bytewise atomically we use an idiom often used with sequence locks~\cite{Hemminger02,Lameter05,Boehm12,Sullivan17}.   Until C++20 there was no effective way to implement this so that it had fully defined behavior~\cite{Boehm12,Boehm20}.   Since C++20 \lstinline{std::atomic_ref} can be used, although for efficiency this makes for complicated code with many special cases.   We hope that C++ adopts the proposal for bytewise atomic loads and stores like that suggested by Hans Boehm~\cite{Boehm20}.

\section{Conclusion}\label{sec:conclusion}

In this paper, we have presented \ustm{}, a simple STM system that achieves both usability and generality while maintaining state-of-the-art performance. To avoid memory leaks and other errors arising from adapting existing sequential code, we introduced and implemented \textit{deferred} aborts, ensuring that user code is never aborted during the speculative phase of a transaction. To reduce the heartbeat of the central clock while maintaining safety for deferred aborts, we introduced the concept of \textit{split-increment timestamps}. We then argued that our algorithm, including deferred aborts, provides opacity---the gold standard of correctness for STMs. Finally, we demonstrated that \ustm{} meets or exceeds the performance of state-of-the-art single-versioned and multiversioned systems.

\begin{acks}

This work was supported in part by the National Science Foundation grant CCF-2119352, and a 
 gift from Jane Street.
 Some experiments presented in this paper were carried out using the Grid'5000 testbed, supported by a scientific interest group hosted by Inria and including CNRS, RENATER and several Universities as well as other organizations (see https://www.grid5000.fr).

\end{acks}

%\balance	
	\bibliographystyle{ACM-Reference-Format}
	% \bibliography{strings,biblio,transactions,references}
  \bibliography{bibliography/strings.bib,bibliography/main.bib,references.bib}

\begin{appendices}

\section{Privatization Safety\label{app:privsafety}}

An issue not commonly addressed by other STM systems is the interoperability of transactional and non-transactional code.
In particular, user programs might require that variables previously accessed inside a transaction be used outside of the STM context (privatization) or vice versa (publication).
The question is how can an STM provide privatization (and publication) safety.

Figure~\ref{codeexample:privatization-safety} illustrates what can go wrong if an STM system does not ensure privatization safety.
The \texttt{is\_private} variable represents whether \texttt{x} can be accessed transactionally (false), or non-transactionally (true).
Thread 1 executes a transaction that sets \texttt{is\_private} (lines 4-6), supposedly enabling a safe raw access to \texttt{x} (line 7).
The issue is that, without correct privatization, the assertion on line 8 can fail.
The pattern, known as \textit{delayed commit}, happens because Thread 2 buffered its write to \texttt{x} (very common amongst STMs) and later overwrote Thread 1's write (line 8).

\hfill
\begin{figure}[h]
\centering
\noindent
\begin{minipage}[t]{0.5\linewidth}
\begin{Verbatim}[frame=lines,numbers=none,label=Thread 1,framesep=1.5mm,xrightmargin=5.9mm]
//is_private == false


transaction([&] {
 is_private.store(true);
});
x = 7;

assert(x == 7); //Fails
\end{Verbatim}
\end{minipage}%
\begin{minipage}[t]{0.5\linewidth}
\begin{Verbatim}[frame=lines,numbers=left,xleftmargin=3mm,label=Thread 2,framesep=1.5mm]
transaction([&] {
 if (!is_private.load())
  x.store(5); //buffered
 //Validate



 //Write to x and Commit
});
\end{Verbatim}
\end{minipage}
\caption{Unsafe interleaving mixing transactional and non-transactional data accesses.}
\label{codeexample:privatization-safety}
\end{figure}

The question is how do we know when it is safe for non-transactional accesses to occur.
%possible solutions
%       non-transactional accesses use stm::load/store
One possibility would be to only access variables that may be accessed transactionally through stm::load and stm::store operations, be it inside or outside transactions.
%               compiler can provide instrumentation non-transactional accesses
This can be done explicitly or automatically through compiler instrumentation\Andre{check if gcc does this}.
%               performance loss
The downside to this approach is the loss of performance accompanied by requiring STM routines to run on all accesses to data that may be used transactionally.
%       raw accesses
%               fence
Another option is to provide an additional \texttt{fence} instruction that allows for privatization (and another for publication\Andre{does publication need a fence?}), as described in~\cite{safeprivatization}.
By requiring that users explicitly fence on privatization, we allow raw accesses to data previously accessed transactionally.
In the example depicted in Figure~\ref{codeexample:privatization-safety}, a \texttt{fence} performed by Thread 1 after line 6 would suffice.
Note that most STMs do not address privatization safety, which forces applications to either only access data inside of transactions or write their own quiescence mechanism to wait for all active transactions to finish.

The \texttt{fence} instruction resembles a fence in the C/C++ memory model.
This is because the raw access in Thread 1 can be seen as analogous to a relaxed access that is not synchronized with Thread 2.
Hence, a fence before the write to \texttt{x} on line 7 can be thought of as preventing reordering of the following raw accesses.
Data-race free semantics are required for correct usage of a fencing solution (see~\cite{safeprivatization} for more detail).

To implement the described fence instruction, we would require a barrier that waits for all currently running transactions to finish, similar to RCU patterns.
The issue then becomes that such an instruction may be too coarse-grained for some use cases.
For example, in Figure~\ref{codeexample:privatization-safety}, one can imagine that if other threads were running and operating on a disjoint set of data, Thread 1 would not need to wait for those threads to finish.
Thus, we propose a \texttt{fence(x)} instruction that only fences on a given object \texttt{x}.
The \texttt{fence(x)} instruction can be implemented much more efficiently, as can be seen in Figure~\ref{fig:alg:fencex}.
By performing a no-side-effects write on \texttt{x}, we ensure that after successfully fencing, currently running transactions will either: have finished performing their write-back phase; or will have aborted, in which case they must be ordered after the fence (e.g., in our example Thread 2 must read false when it reads \texttt{is\_private}).

\begin{figure}
\begin{Verbatim}[frame=lines,numbers=left,xleftmargin=3mm]
  template <typename T> inline void fence(T& d) {
    transaction([=] {d.store(d.load());}); }
\end{Verbatim}
\caption{\texttt{fence(x)} implementation.}
\label{fig:alg:fencex}
\end{figure}

%\section{Privatization Safety}---------

\newpage

\section{\ustm{} Code}
\label{sec:ustmcode}

\makeatletter
 \lstset{
    language=C++,
    basicstyle=\footnotesize\ttfamily, 
    numbers=left, 
    numberstyle=\tiny\color{gray},
    frame=single,
    literate={<}{<}1
            {>}{>}1
}

\begin{lstlisting}
// Supports transactions with the following interface:
//   r = ustm::transaction(func) // runs a func in transaction
//   r = ustm::transaction(func, true) // read only mode
//   r = ustm::load(x) // loads value from x
//   ustm::store(x, v) // stores v into x
//   r = ustm::New<T>(args ...) // new object of type T
//   ustm::Delete(x) // deletes x
//   ustm::fence(x) // for privatization
//   ustm::fence() // for privatization
// The types for load and store can be any relocatable type.

#ifndef USTM_H_
#define USTM_H_
#include <array>
#include <atomic>
#include <bitset>
#include <cmath>
#include <functional>
#include <iostream>
#include <thread>
#include <tuple>
#include <unordered_map>
#include <vector>
// Uses epoch-based SMR from uepoch.h.  Needed interface:
//   r = uepoch::protect(func) : runs func under protection
//   uepoch::delay(func) : delays until protected are done
//   uepoch::add_before_epoch_hook(func)
//   uepoch::thread_id()
//   uepoch::quiesce() : waits for all protected regions
#include "uepoch.h"

namespace ustm {
  // some constants
  constexpr int MaxThreads = 4096; // can be increased
  constexpr int logLocks = 22; // log_2 of number of locks
  static constexpr int filterBits = 10; // for hash filter

  using TS = size_t; // type for timestamps

  struct timeStamp {
    std::atomic<TS> ts = 1ul;
    // returns current stamp
    TS getStamp() const { return ts.load(); }
    // returns incremented stamp
    TS nextStamp(TS prevStamp=0) {
      if (prevStamp == 0) prevStamp = getStamp();
      TS stamp = ts.load(); 
      if (stamp > prevStamp) return stamp;
      for (volatile int i = 0; i < 200; i++);
      stamp = ts.load(); 
      if (stamp > prevStamp) return stamp;
      return ++ts;
    }
  };

  // generic version link
  struct version {
    version* next = nullptr;
    TS stamp = 0;
    void* addr;
    version(void* addr) : addr(addr) {}
  };

  // link including a type specific value
  template <typename T>
  struct link : version {
    T value;
    T scratch;
    link(T* addr, T value) : version{addr}, value(value) {}
  };
  
  // The lock structure.  Contains the status of the lock
  // and if using versioning a version list.
  struct alignas(16) lock {
    using Status = size_t;
    // If locked (high bit set) keeps tid of owner,
    // otherwise keeps timestamp
    std::atomic<Status> v = 1; // time starts at 1
    lock() {}
    static int getTID(Status s) { return ~(1ul << 63) & s;}
    static TS getStamp(Status s) { return s;}
    static bool isLocked(Status s) { return (s >> 63) & 1ul;}
    static Status setLock(int tid) {
      return Status((1ul << 63) | tid);}
    Status getStatus() const {
      return v.load(std::memory_order_acquire);}
    void setStamp(TS stamp) {
      v.store(stamp, std::memory_order_release); }

    std::atomic<version*> verlist = nullptr; // version list
    version* getVersionList() const { return verlist.load();}
    
    // add to version list
    void addLink(version* l, TS prevStamp) {
      l->next = verlist.load();
      l->stamp = prevStamp;
      verlist.store(l, std::memory_order_release); }

    // Returns whether succeeded or not, and if succeeded a
    // 0 stamp.  If was self locked, and the swapped out
    // timestamp otherwise
    std::pair<bool,TS> tryLock(int tid, const TS startStamp) {
      Status s = getStatus();
      if (isLocked(s)) 
        if (getTID(s) == tid) return std::pair(true, 0);
        else return std::pair(false, 0);
      if (s < startStamp &&
          v.compare_exchange_strong(s, setLock(tid))) {
        return std::pair(true, getStamp(s));
      } else return std::pair(false, 0);
    }

    // unlocks by setting status to timestamp
    void unlock(const TS ts) {
      v.store(ts, std::memory_order_release); }

    // checks that either self locked or timestamp before
    // startStamp
    bool validateRead(const TS startStamp, int tid) const {
      Status s = getStatus();
      return ((isLocked(s) && getTID(s) == tid) ||
              s < startStamp); 
    }
  };

  struct STMState;
  
  // Transaction descriptor, includes the logs, stats, and
  // other state
  struct alignas(128) Descriptor {
    bool inTransaction {false};
    bool readOnly {false};
    bool lateRead {false};
    int inConstructor {0};
    size_t tid; // thread identifier
    size_t numRetries {0};    
    TS startStamp;
    STMState* stmState;
    
    // Entry for the write log.  Contains the location, an
    // oldStamp used when locked, a pointer to the link
    // containing the value, a pointer to the data within
    // the link and the size in bytes of the value.
    struct logEntry {
      void* loc;
      TS oldStamp;
      version* link;
      void* data;
      void* scratch;
      int size;
    };
    
    // The following four members are the logs.
    std::vector<logEntry> writeLog;
    // The read log contains pointers to read locations.
    std::vector<const void*> readLog;

    // Alloc and delete logs keep pointers to functions that
    // do deletes
    std::vector<std::function<void()>> allocLog;
    std::vector<std::function<void()>> deleteLog;

    // A hash map and filter to let reads find prior writes
    std::unordered_map<void*, int> writeMap;
    std::bitset<filterBits> writeFilter;

    // used to reset between each attempt
    void reset() {
      if (writeLog.size() > 0) {
        writeFilter = 0; writeLog.clear(); writeMap.clear();
      }
      readLog.clear(); deleteLog.clear(); allocLog.clear();
      inConstructor = false; lateRead = false;
    }
  };

  template <int bits>
  size_t hash(const void* addr) {
    size_t x = ((reinterpret_cast<size_t>(addr)>>6) *
                UINT64_C(0xbf58476d1ce4e5b9));
    return x >> (64 - bits);
  }
  
  struct STMState {
    // state consists of locks, descriptors, and timestamp
    static constexpr size_t numLocks = 1ul << logLocks;
    std::array<Descriptor,MaxThreads> descriptors;
    std::array<lock,numLocks> writeLocks;
    timeStamp ts;

    // Returns a lock corresponding to a hash of the address
    inline lock& getLock(const void* addr) {
      return writeLocks[hash<logLocks>(addr)];} 

    // returns true if successful
    inline bool commitTransaction(Descriptor* myd) {
      TS endStamp;
      bool hasWrite = myd->writeLog.size() > 0;
      if (hasWrite) {
        if (myd->lateRead) return false;
        // Take locks (abort if any fail)
        for (auto& e : myd->writeLog) {
          auto [r,s] = getLock(e.loc).tryLock(myd->tid,
                                              myd->startStamp);
          if (!r) return false;
          e.oldStamp = s;
        }
        // get stamp -- this is the serialization point
        endStamp = ts.getStamp();
        // Validate the reads (abort if any fail)
        for (auto& e : myd->readLog) 
          if (!getLock(e).validateRead(myd->startStamp,
                                       myd->tid))
            return false;
        // If  here (i.e.  locks  and validates  succeeded),
        // transaction succeeded.  Apply the writes and then
        // unlock them.
        for (auto& e : myd->writeLog) {

          auto *tmp = e.scratch;
          // swap destination and data (from the link)
          std::memcpy(tmp, e.loc, e.size);
          std::memcpy(e.loc, e.data, e.size);
          std::memcpy(e.data, tmp, e.size);
          // add link to version list
          TS writeStamp = ((e.oldStamp == 0)
                           ? endStamp : e.oldStamp);
          getLock(e.loc).addLink(e.link, writeStamp);
        }
        std::atomic_thread_fence(std::memory_order_release);
        for (auto&  e : myd->writeLog) // unlock writes
         if (e.oldStamp != 0)
           getLock(e.loc).unlock(endStamp);
      }
      // Since successful, apply the deletes.
      for (auto& e : myd->deleteLog)
        uepoch::delay(std::move(e));
      myd->reset();
      if (hasWrite) ts.nextStamp(endStamp);
      myd->inTransaction = false;
      return true;
    }

    inline void abortTransaction(Descriptor* myd) {
      // Release locks and apply deletes to allocated objects.
      for (auto& e : myd->writeLog) 
        if (e.oldStamp != 0) getLock(e.loc).unlock(e.oldStamp);
      for (auto& e : myd->allocLog)
        uepoch::delay(std::move(e));
      if (++myd->numRetries % 4000000 == 0) { 
        std::cout << "ustm: too many retries: "
                  << std::endl; abort();}
      myd->reset();
      ts.nextStamp(myd->startStamp);
    }
    
    STMState() {
      for (int i=0; i < MaxThreads; i++) {
        descriptors[i].tid = i;
        descriptors[i].stmState = this;
      }
      // ensures stamp is incremented whenever epoch is
      uepoch::add_before_epoch_hook([&] { ts.nextStamp(); });
    }
    ~STMState() {}
  };
      
  extern inline Descriptor* initDescriptor() {
    static STMState stm;
    return &stm.descriptors[uepoch::thread_id()];
  }

  // thread local copy of descriptor.
  extern inline Descriptor* getDescriptor() {
    static thread_local Descriptor* d = initDescriptor();
    return d;
  }

  template <typename T>
  inline T load(T& v, bool unprotected=false) {
    Descriptor* myd = getDescriptor();
    lock& lck = myd->stmState->getLock(&v);
    auto status = lck.getStatus();

    // check if value is in the write buffer
    if (myd->writeLog.size() > 0 &&
        myd->writeFilter[hash<filterBits>(&v)]) {
      auto a = myd->writeMap.find((void*) &v);
      if (a != myd->writeMap.end()) {
        version* ptr = (myd->writeLog[(*a).second]).link;
        return (reinterpret_cast<link<T>*>(ptr))->value;
      }
    }

    version* nxt;
    bool regular = (myd->inTransaction &&
                    !myd->readOnly && !unprotected);
    if (regular) myd->readLog.push_back(&v);
    char tmp[sizeof(T)];
    T* location = reinterpret_cast<T*>(tmp);
    while (true) { 
      auto prevStatus = status;
      std::memcpy(tmp, &v, sizeof(T));
      nxt = lck.getVersionList();
      std::atomic_thread_fence(std::memory_order_acquire);
      status = lck.getStatus();
      if (prevStatus == status) {
        if (status < myd->startStamp) return *location;
        if (!lock::isLocked(status)) {
          if (!myd->inTransaction) return *location;
          break;
        }
      }
    }
    if (regular) myd->lateRead = true;
    TS stamp = lck.getStamp(status);
    // Chase down the version list to the right one
    while (stamp >= myd->startStamp && nxt != nullptr) {
      if (nxt->addr == &v)
        location = &(reinterpret_cast<link<T>*>(nxt)->value);
      stamp = nxt->stamp;
      nxt = nxt->next;
    }
    return *location;
  }

  // runs func in an transaction
  // func takes no argument and must return a value
  template<typename F>
  inline auto transaction(F&& func, bool readOnly = false) {
    Descriptor* myd = getDescriptor();
    STMState* stm = myd->stmState;
    // if nested then just run
    if (myd->inTransaction) return func(); 
    myd->inTransaction = true;
    myd->numRetries = 0;
    myd->readOnly = readOnly;
    while (true) {
      auto returnValue = uepoch::protect([&] {
        myd->startStamp = stm->ts.getStamp();
        return func();
      });
      if (stm->commitTransaction(myd))
        return returnValue;
      else stm->abortTransaction(myd);
    }
  } 

  // adds to log so can delete if aborted
  template <typename T, typename... Args>
  inline T* New(Args&&... args) {
    Descriptor* myd = getDescriptor();
    if (myd->inTransaction) {
      myd->inConstructor++;
      T* ptr = new T(args...);
      myd->allocLog.push_back([=] {delete ptr;});
      myd->inConstructor--;
      return ptr;
    } else
      return new T(args...);
  }

  // adds to log and only applies deletes on success
  template<typename T>
  inline void Delete(T* obj) {
    if (obj == nullptr) return;
    Descriptor* myd = getDescriptor();
    if (!myd->inTransaction) delete obj;
    else myd->deleteLog.push_back([=] { delete obj;});
  }

  template <typename T>
  inline void store(T& loc, const T& v) {
    Descriptor* myd = getDescriptor();
    if (myd->inConstructor > 0) loc = v;
    else if (!myd->inTransaction)
      transaction([&] {ustm::store(loc, v); return true;});
    else {
      // insert into hash filter and map so reads can find it
      myd->writeFilter[hash<filterBits>(&loc)] = 1;
      myd->writeMap[&loc] = myd->writeLog.size();
      // temporarily store value in new link
      auto* hp = ustm::New<link<T>>(&loc, v);
      ustm::Delete(hp); // delete is delayed
      myd->writeLog.push_back({
        .loc = &loc,
        .oldStamp = 0,
        .link = hp,
        .data = &hp->value,
        .scratch = &hp->scratch,
        .size = sizeof(T)
      });
    } 
  }

  template <typename T>
  inline void fence(T& loc) {
    ustm::transaction([=] {
       ustm::store(loc, ustm::load(loc));}); }

  inline void fence() { uepoch::quiesce();}
} // End namespace ustm

#endif // USTM_H_
\end{lstlisting}

\end{appendices}

\end{document}